\documentclass[sn-mathphys]{sn-jnl}% Math and Physical Sciences Reference Style
\usepackage{amsthm}
\usepackage{amsmath}
\usepackage{amssymb}
\usepackage{natbib}
\usepackage{hyperref}
\usepackage{graphicx}
\usepackage{caption}
\usepackage{subcaption}
\usepackage{cleveref}

\jyear{2021}%

\theoremstyle{thmstyleone}%
\newtheorem{theorem}{Theorem}%  meant for continuous numbers
\newtheorem{proposition}[theorem]{Proposition}%
\newtheorem{lemma}[theorem]{Lemma}% 
\newtheorem{corollary}[theorem]{Corollary}% 

\theoremstyle{thmstyletwo}%
\newtheorem{remark}{Remark}%

\theoremstyle{thmstylethree}%
\newtheorem{definition}{Definition}%

\renewcommand{\R}{\mathbb{R}}
\newcommand{\paren}[1]{\left({#1}\right)}
\renewcommand{\set}[1]{\left\{{#1}\right\}}
\newcommand{\abs}[1]{\vert {#1}\vert}

\begin{document}

\title[Modified Riemannian Manifold and Lagrangian Monte Carlo]{Geometric Ergodicity in Modified Variations of Riemannian Manifold and Lagrangian Monte Carlo}


\author*[1]{\fnm{James A.} \sur{Brofos}}\email{james.brofos@yale.edu}

\author[2]{\fnm{Vivekananda} \sur{Roy}}\email{vroy@iastate.edu}

\author[1]{\fnm{Roy R.} \sur{Lederman}}\email{roy.lederman@yale.edu}

\affil[1]{\orgdiv{Department of Statistics and Data Science}, \orgname{Yale University}, \orgaddress{\street{24 Hillhouse Ave.}, \city{New Haven}, \postcode{06510}, \state{Connecticut}, \country{USA}}}

\affil[2]{\orgdiv{Department of Statistics}, \orgname{Iowa State University}, \orgaddress{\street{2438 Osborn Drive}, \city{Ames}, \postcode{50011}, \state{Iowa}, \country{USA}}}

% \affil[3]{\orgdiv{Department}, \orgname{Organization}, \orgaddress{\street{Street}, \city{City}, \postcode{610101}, \state{State}, \country{Country}}}


\abstract{Riemannian manifold Hamiltonian (RMHMC) and Lagrangian Monte Carlo (LMC) have emerged as powerful methods of Bayesian inference. Unlike Euclidean Hamiltonian Monte Carlo (EHMC) and the Metropolis-adjusted Langevin algorithm (MALA), the geometric ergodicity of these Riemannian algorithms has not been extensively studied. On the other hand, the manifold Metropolis-adjusted Langevin algorithm (MMALA) has recently been shown to exhibit geometric ergodicity under certain conditions. This work investigates the mixture of the LMC and RMHMC transition kernels with MMALA in order to equip the resulting method with an ``inherited'' geometric ergodicity theory. We motivate this mixture kernel based on an analogy between single-step HMC and MALA. We then proceed to evaluate the original and modified transition kernels on several benchmark Bayesian inference tasks.}




\maketitle

\section{Introduction}

Bayesian inference seeks to combine subjective sources of information with observational data. By specifying one's prior beliefs and correctly capturing sources of uncertainty in a stochastic system, one may employ the Bayesian approach in order to capture and reason about uncertainty under the posterior distribution. Formally, the posterior distribution is a probability distribution $\Pi:\mathfrak{B}(\R^{m})\to [0,1]$ with density $\pi:\R^m\to\R_{+}$. As a practical matter in Bayesian inference, one is concerned with the computation of expectations with respect to $\Pi$ of integrable functions $h:\R^m\to\R$; that is, we wish to compute $\underset{q\sim \Pi}{\mathbb{E}}[h(q)]$. Since $\pi$ is generally intractable, these expectations are not available in closed form and Monte Carlo methods based on samples from $\pi$ are often used to approximate these means. For instance, one may compute the mean and variance of the posterior by employing appropriate choices of expectation. Therefore, a central problem in Bayesian inference is the generation of samples from the posterior distribution.

One popular Monte Carlo method is based on approximate samples from the posterior distribution generated via the technique of
Markov chain Monte Carlo (MCMC). In MCMC, given some initial state of the chain, $q^0$, a sequence of random variables $(q^1, q^2, \ldots)$ is generated inductively: the conditional distribution $q^n \vert q^{n-1},\ldots, q^0$ obeys the Markov property. Under certain desirable regularity conditions, the random variables $q^n$ will have a distribution which, asymptotically, converges to $\Pi$. Establishing the {\it rate} of convergence, or at least an upper bound on that rate, serves to quantify how quickly the chain $(q^1, q^2,\ldots)$ mixes toward the target distribution. In the case where the rate of convergence is geometrically fast, the Markov chain enjoys a strong stability property in that the estimator $n^{-1} \sum_{i=1}^n h(q^i)$ of $\underset{q\sim \Pi}{\mathbb{E}} [h(q)]$ can be equipped with a central limit theorem under additional technical assumptions. This central limit theorem has important practical implications in that it allows the use of asymptotically valid standard errors of MCMC estimates, which, in turn, can be used to decide how long to run the Markov chains \citep{roy:2020}.

One of the most popular MCMC methods is Euclidean Hamiltonian Monte Carlo (EHMC) \citep{neal-hmc}. In EHMC, one computes approximate solutions to Hamilton's equations of motion using numerical integrators, taking care to ensure that the resulting Markov chain satisfies detailed balance in a {\it phase-space} consisting of the {\it position} variables $q\in \R^m$ and auxiliary {\it momentum} variables $p\in\R^m$. The rate of convergence of HMC was established by \citet{geometric-ergodicity-hmc} in the presence of conditions on the log-density \citep[see also][]{durm:moul:saks:2020}. Traditional EHMC, however, struggles in distributions that exhibit multiple spatial scales \citep{pourzanjani2019implicit,softabs}. This observation led to the development of geometric methods of HMC, specifically the Riemannian manifold Hamiltonian Monte Carlo (RMHMC) method of \citet{rmhmc}.

Unlike EHMC, RMHMC adapts to the second-order structure of the posterior, which allows it to align its proposals in the direction of the posterior that exhibits the greatest local variation. However, the sophisticated form of the Hamiltonian employed in RMHMC necessitates the use of complex numerical integrators that are significantly more expensive than the numerical integrator employed in EHMC. This concern was partially alleviated in the introduction of Lagrangian Monte Carlo (LMC) in \citet{Lan_2015}, which was able to construct a MCMC method with a more efficient numerical integration procedure, while simultaneously continuing to take advantage of second-order geometric knowledge in order to efficiently traverse the posterior.  Both RMHMC and LMC can be viewed as generalizations of the Euclidean Hamiltonian Monte Carlo algorithm, which incorporate second-order geometric information about the posterior into the Markov chain transition kernel. Incorporating second-order information allows these geometric methods of MCMC to explore the typical region of the target distribution more efficiently.

However, neither RMHMC nor LMC have, as of yet, been equipped with a geometric ergodicity theory. The focus in this work is in establishing the geometric ergodicity of modified versions of RMHMC and LMC. In implementations of RMHMC and LMC it is common to randomize the number of integration steps, and to take one-step with positive probability. In the case of Euclidean HMC, the equivalence between single-step HMC and the Metropolis-adjusted Langevin algorithm (MALA) is critical for establishing geometric ergodicity of HMC in \citet{geometric-ergodicity-hmc}. Our approach is two-fold. First, we propose a simple modification to both RMHMC and LMC which is motivated by a unique correspondence between single-step Euclidean HMC and MALA. In particular, we propose that instead of applying the transition kernel corresponding to RMHMC or LMC with a single integration step, one instead applies the transition kernel of the {\it manifold} Metropolis-adjusted Langevin algorithm (MMALA). Recent work in \citet{roy2021convergence} gave conditions under which MMALA is geometrically ergodic. We may imbue these modified variations of RMHMC and LMC with a geometric ergodicity theory via inheritance once one establishes that the RMHMC and LMC transition kernels (with a fixed, non-random number of integration steps) are reversible in the required sense. This construction is described in \cref{geometric-ergodicity:subsec:inherited-geometric-ergodicity}. In \cref{geometric-ergodicity:sec:experiments} we proceed to a numerical evaluation of the proposed modification of RMHMC and LMC, with special attention given to the probability of applying the MMALA transition kernel. We begin, however, in \cref{geometric-ergodicity:sec:preliminaries} with an overview of required mathematical concepts from numerical integration and Markov chains.

\section{Preliminaries}\label{geometric-ergodicity:sec:preliminaries}

In \cref{geometric-ergodicity:subsec:numerical-integrators} we review the generalized leapfrog and Lagrangian leapfrog employed in geometric MCMC methods for performing Bayesian inference. \Cref{geometric-ergodicity:subsec:markov-chains} covers Markov chains based on involutions as well as those based on discretizations of Langevin diffusions. We also discuss geometric ergodicity, mixture transition kernels, and the equivalence between single-step HMC and MALA in the Euclidean regime.

\subsection{Numerical Integrators}\label{geometric-ergodicity:subsec:numerical-integrators}

Hamiltonian and Lagrangian mechanics are equivalent, with the momentum and velocity being related by the Legendre transform \citep{mechanics-and-symmetry}. In this section, we will consider numerical integrators of these equations of motion. Throughout our discussion, we will consider a Hamiltonian function of the following form.
\begin{definition}\label{geometric-ergodicity:def:riemannian-hamiltonian}
Let $\pi:\R^m\to\R_+$ be a smooth probability density and let $\mathbf{G} : \R^m\to\mathrm{PD}(m)$. The {\it Riemannian Hamiltonian} is the function $H : \R^m\times\R^m\to\R$ defined by
\begin{align}
    \label{geometric-ergodicity:eq:riemannian-hamiltonian} H(q, p) = -\log\pi(q) + \frac{1}{2} \log\mathrm{det}(\mathbf{G}(q)) + \frac{1}{2} p^\top \mathbf{G}^{-1}(q) p.
\end{align}
Moreover, assuming $\int_{\R^m}\int_{\R^m} \exp(-H(q, p))~\mathrm{d}q~\mathrm{d}p < +\infty$, the {\it Riemannian density} is the probability density $\pi(q, p)\propto \exp(-H(q, p))$.
\end{definition}
\begin{remark}
The term $\frac{1}{2} \log\mathrm{det}(\mathbf{G}(q)) + \frac{1}{2} p^\top \mathbf{G}^{-1}(q) p$ appearing in the definition of the function $H$ in \cref{geometric-ergodicity:eq:riemannian-hamiltonian} is chosen such that the conditional density $\pi(q, p)$ satisfies
\begin{align}
    \pi(p\vert q) = \mathrm{Normal}(0, \mathbf{G}(q)).
\end{align}
It is easily seen that the $q$-marginal density of $\pi(q, p)$ obtained by marginalizing out $p$ is $\pi(q)$.
\end{remark}
Hamiltonian functions $H:\R^m\times\R^m\to \R$ produce Hamilton's equations of motion which, by definition, are solutions to the coupled differential equations
\begin{align}
    \label{geometric-ergodicity:eq:position-evolution} \dot{q}_t &= \nabla_p H(q_t, p_t) \\ 
    \label{geometric-ergodicity:eq:momentum-evolution} \dot{p}_t &= -\nabla_q H(q_t, p_t).
\end{align}
Except in special cases, a closed-form solution for the map $t\mapsto (q_t,p_t)\in\R^m\times\R^m$ does not exist; this necessitates the use of numerical integrators in order to approximate the equations of motion obeying \cref{geometric-ergodicity:eq:position-evolution,geometric-ergodicity:eq:momentum-evolution}. One such example is the generalized leapfrog method, which we now define.
\begin{definition}\label{geometric-ergodicity:def:riemannian-leapfrog}
The {\it generalized leapfrog integrator} with step-size $\epsilon\in\R$ applied to the Riemannian Hamiltonian in \cref{geometric-ergodicity:eq:riemannian-hamiltonian} is the map $(q, p)\mapsto \hat{\Phi}_\epsilon(q, p)$ defined by,
\begin{align}
    \label{geometric-ergodicity:eq:intermediate-momentum} \breve{p} &= p -\frac{\epsilon}{2} \nabla_q H(q, \breve{p}) \\
    \tilde{q} &= q + \frac{\epsilon}{2} \paren{\nabla_p H(q, \breve{p}) + \nabla_p H(\tilde{q}, \breve{p})} \\
    \label{geometric-ergodicity:eq:position} &= q + \frac{\epsilon}{2} \paren{\mathbf{G}^{-1}(q) + \mathbf{G}^{-1}(\tilde{q})}\breve{p} \\
    \tilde{p} &= \breve{p} - \frac{\epsilon}{2} \nabla_q H(\tilde{q}, \breve{p}) \\
    &= p -\frac{\epsilon}{2} \nabla_q H(q, \breve{p}) - \frac{\epsilon}{2} \nabla_q H(\tilde{q}, \breve{p}),
\end{align}
where $\hat{\Phi}_\epsilon(q, p) = (\tilde{q}, \tilde{p})$.
\end{definition}
\begin{remark}
One can prove \citep{leimkuhler_reich_2005,hairer-geometric} that the generalized leapfrog integrator is a symplectic transformation and that therefore $\abs{\mathrm{det}(\nabla\Phi(q, p))} = 1$.
\end{remark}
\begin{remark}
We must now give a remark about notation. Numerical integration plays a central role in our analysis, and we give particular attention to multiple steps of integration. At the same time, we will also discuss Markov chains, which also consist of multiple steps. As an attempt to differentiate these two notions of step, we will use a lower index to refer to steps of numerical integration, whereas we will use upper indices to denote Markov chain steps.
Given initial data $(q_0, p_0)\in\R^m\times\R^m$, we denote $k$-steps of the generalized leapfrog integrator by $(q_k, p_k) = \hat{\Phi}^k_\epsilon(q_0, p_0) = \hat{\Phi}_\epsilon(q_{k-1}, p_{k-1})$. Similarly, we let $\breve{p}_k = p_{k-1} - \frac{\epsilon}{2} \nabla_q H(q_{k-1}, \breve{p}_k)$, which we call the {\it intermediate momentum at the $k$-th step}.
\end{remark}
A Hamiltonian of the form in \cref{geometric-ergodicity:eq:riemannian-hamiltonian} may be transformed into a Lagrangian as
\begin{align}
    L(q, \dot{q}) = -\log \pi(q) + \frac{1}{2} \log\mathrm{det}(\mathbf{G}(q)) - \frac{1}{2} \dot{q}^\top \mathbf{G}(q) \dot{q}.
\end{align}
As noted at the beginning of this subsection, Lagrangian and Hamiltonian mechanics are formally equivalent, with the momentum $p\in\R^m$ being related to the velocity $\dot{q}\in\R^m$ according to the Legendre transformation $p= \mathbf{G}(q) \dot{q}$. The Lagrangian produces equivalent equations of motion called the Euler-Lagrange equations as
\begin{align}
    \nabla_q L(q_t, \dot{q}_t) = \frac{\mathrm{d}}{\mathrm{d}t} \nabla_{\dot{q}} L(q_t, \dot{q}_t).
\end{align}
One advantage possessed by the Lagrangian formalism over the Hamiltonian approach is that one may identify {\it explicit} numerical integrators \citep{Lan_2015} of the equations of motion, such as the Lagrangian leapfrog.
\begin{definition}\label{geometric-ergodicity:def:lagrangian-leapfrog}
The {\it Lagrangian leapfrog integrator} with step-size $\epsilon\in \R$ applied to the Riemannian Hamiltonian in \cref{geometric-ergodicity:eq:riemannian-hamiltonian} is the map $(q, p)\mapsto \tilde{\Phi}_{\epsilon}(q, p)$ defined by,
\begin{align}
    v &= \mathbf{G}^{-1}(q) p \\
    \label{geometric-ergodicity:eq:lagrange-intermediate-velocity} \breve{v} &= v - \frac{\epsilon}{2} \mathbf{G}^{-1}(q) \nabla_q U(q) - \frac{\epsilon}{2} \Omega(q, v)\breve{v} \\
    \label{geometric-ergodicity:eq:lagrange-final-position} \tilde{q} &= q + \epsilon \breve{v} \\
    \tilde{v} &= \breve{v} - \frac{\epsilon}{2} \mathbf{G}^{-1}(\tilde{q}) \nabla_q U(\tilde{q}) - \frac{\epsilon}{2} \Omega(\tilde{q}, \breve{v})\tilde{v} \\
    \label{geometric-ergodicity:eq:lagrange-final-velocity} &= v - \frac{\epsilon}{2} \mathbf{G}^{-1}(q) \nabla_q U(q) - \frac{\epsilon}{2} \Omega(q, v)\breve{v} - \frac{\epsilon}{2} \mathbf{G}^{-1}(\tilde{q}) \nabla_q U(\tilde{q}) - \frac{\epsilon}{2} \Omega(\tilde{q}, \breve{v})\tilde{v} \\
    \tilde{p} &= \mathbf{G}^{-1}(\tilde{q}) \tilde{v}
\end{align}
where $(\tilde{q},\tilde{p}) = \tilde{\Phi}_{\epsilon}(q, p)$ and
\begin{align}
    U(q) &= -\log\pi(q) + \frac{1}{2} \log\mathrm{det}(\mathbf{G}(q)) \\
    \Omega_{ij}(q, v) &= \sum_{k=1}^m \Gamma^i_{kj}(q) v_k \\
    \Gamma^i_{kj}(q) &= \frac{1}{2} \sum_{l=1}^m \mathbf{G}^{-1}_{kl}(q) \paren{\frac{\partial}{\partial q_i} \mathbf{G}_{lj}(q) + \frac{\partial}{\partial q_j} \mathbf{G}_{il}(q) - \frac{\partial}{\partial q_l} \mathbf{G}_{ij}(q)}.
\end{align}
\end{definition}
The Jacobian determinant of the transformation $\tilde{\Phi}_\epsilon$ (defined in \cref{geometric-ergodicity:def:lagrangian-leapfrog}) is,
\begin{align}
    \label{geometric-ergodicity:eq:lagrangian-jacobian-determinant} \abs{\mathrm{det}(\nabla\tilde{\Phi}_\epsilon(q, p))} = \mathrm{det}(\mathbf{G}^{-1}(q))\mathrm{det}(\mathbf{G}(\tilde{q})) \frac{\abs{\mathrm{det}(\mathrm{Id} + \epsilon \Omega(q, v)/2) \mathrm{det}(\mathrm{Id} - \epsilon \Omega(\tilde{q},\tilde{v})/2)}}{\abs{\mathrm{det}(\mathrm{Id} + \epsilon \Omega(\tilde{q}, \bar{v})/2) \mathrm{det}(\mathrm{Id} - \epsilon \Omega(q,\bar{v})/2)}}.
\end{align}
\begin{remark}
Given initial data $(q_0, p_0)\in\R^m\times\R^m$, we denote $k$-steps of the generalized (or Lagrangian) leapfrog integrator by $(q_k, p_k) = \tilde{\Phi}^k_\epsilon(q_0, p_0) = \tilde{\Phi}_\epsilon(q_{k-1}, p_{k-1})$.
\end{remark}

In the case where $\mathbf{G}$ is a constant function of $q$, the  generalized leapfrog integrator reduces to the standard leapfrog integrator, which is defined as follows.
\begin{definition}\label{geometric-ergodicity:def:euclidean-hamiltonian}
Let $\pi:\R^m\to\R_+$ be a smooth probability density. A {\it Euclidean Hamiltonian} is a smooth map $H:\R^m\times\R^m\to\R$ of the form,
\begin{align}
    \label{geometric-ergodicity:eq:euclidean-hamiltonian} H(q, p) = -\log \pi(q) + \frac{1}{2} p^\top \mathbf{G}^{-1} p,
\end{align}
where $\mathbf{G}\in\mathrm{PD}(m)$.
\end{definition}
\begin{definition}\label{geometric-ergodicity:def:euclidean-leapfrog}
Consider a Euclidean Hamiltonian as defined in \cref{geometric-ergodicity:def:euclidean-hamiltonian}. The {\it Euclidean leapfrog integrator} is the map $(q, p)\mapsto \check{\Phi}_\epsilon(q, p)$ defined by
\begin{align}
    \label{geometric-ergodicity:eq:euclidean-intermediate-momentum} \breve{p} &= p + \frac{\epsilon}{2} \nabla \log\pi(q) \\
    \tilde{q} &= q + \epsilon \mathbf{G}^{-1}\breve{p} \\
    \label{geometric-ergodicity:eq:euclidean-final-momentum} \tilde{p} &= \breve{p} + \frac{\epsilon}{2} \nabla \log\pi(\tilde{q}),
\end{align}
where $\check{\Phi}_\epsilon(q, p) = (\tilde{q}, \tilde{p})$.
\end{definition}
The Euclidean leapfrog integrator is the {\it de-facto} standard numerical integrator employed in EHMC. This is because of its accuracy and computational efficiency; in contrast to the generalized leapfrog integrator in \cref{geometric-ergodicity:def:riemannian-leapfrog}, the Euclidean leapfrog is fully explicit. Nevertheless, the generalized leapfrog and Euclidean leapfrog are exactly equivalent when applied to a Hamiltonian in the form of \cref{geometric-ergodicity:def:euclidean-hamiltonian}; this is made precise in the following result.
\begin{proposition}\label{geometric-ergodicity:prop:generalized-langrangian-euclidean-equivalent}
When the generalized leapfrog (\cref{geometric-ergodicity:def:riemannian-leapfrog}) or Lagrangian leapfrog (\cref{geometric-ergodicity:def:lagrangian-leapfrog}) is applied to a Euclidean Hamiltonian in the form of \cref{geometric-ergodicity:def:euclidean-hamiltonian} (in the sense that we take $\mathbf{G}(q) \equiv \mathbf{G}$ for every $q\in\R^m$), the resulting map is equivalent to the Euclidean leapfrog method in \cref{geometric-ergodicity:def:euclidean-leapfrog}.
\end{proposition}
\begin{proof}
In the case of the generalized leapfrog integrator, this is immediate from the fact that $\nabla_q H(q, p) = -\nabla \log\pi(q)$ which does not depend on $p$ and $\nabla_p H(q, p) = \mathbf{G}p$ which does not depend on $q$.

In the case of the Lagrangian leapfrog, the form of the Euclidean Hamiltonian means that $\Omega(q, v)$ is uniformly zero (because all of the derivatives of the metric vanish). The Lagrangian leapfrog integrator may be re-expressed as
\begin{align}
    v &= \mathbf{G}^{-1} p \\
    \breve{v} &= v + \frac{\epsilon}{2} \mathbf{G}^{-1} \nabla \log \pi(q) \\
    &= \mathbf{G}^{-1} \breve{p} \\
    \tilde{q} &= q + \epsilon \breve{v} \\
    &= q + \epsilon \mathbf{G}^{-1} \breve{p} \\
    \tilde{v} &= \breve{v} + \frac{\epsilon}{2} \mathbf{G}^{-1}\nabla \log\pi(\tilde{q}) \\
    &= \mathbf{G}^{-1} \tilde{p},
\end{align}
where $\breve{p}$ and $\tilde{p}$ are as in \cref{geometric-ergodicity:eq:euclidean-intermediate-momentum,geometric-ergodicity:eq:euclidean-final-momentum}.
\end{proof}

\subsection{Markov Chains}\label{geometric-ergodicity:subsec:markov-chains}

The purpose of this section is to review fundamentals of Markov chains, including their construction, convergence properties, and central limit theorems.
Markov chains are stochastic processes that may be inductively defined by their transition kernel. Since the transition kernel yields a probability measure depending only on the current state, one sees by inspection that the Markov chain satisfies the Markov property.
\begin{definition}
A Markov chain transition kernel on $\R^m$ is a function $K : \R^m\times\mathfrak{B}(\R^m) \to [0, 1]$ such that (i) for $q\in \R^m$, $K(q, \cdot)$ is a probability measure and (ii) for fixed $A\in\mathfrak{B}(\R^m)$ the function $q\mapsto K(q, A)$ is measurable.
\end{definition}
Within the context of simulation-based inference and mainly Bayesian inference, one is interested in constructing a Markov chain which converges to a specified target probability distribution; as noted in the introduction, this target distribution is typically known by its density up to a normalizing constant. We now define two notions of convergence.
\begin{definition}
Given a Markov chain transition kernel $K$, a Markov chain is a sequence of random variables defined inductively by $q^{n+1}\vert q^n \sim K(q^n, \cdot)$.
\end{definition}
\begin{definition}
The {\it total variation distance} between two probability measures $\Pi : \mathfrak{B}(\R^m)\to [0, 1]$ and $\Xi : \mathfrak{B}(\R^m)\to[0,1]$ is defined by
\begin{align}
    \Vert \Pi-\Xi\Vert_{\mathrm{TV}} = 2 \sup_{A\in\mathfrak{B}(\R^m)} \abs{\Pi(A)-\Xi(A)}.
\end{align}
\end{definition}
\begin{definition}\label{geometric-ergodicity:def:ergodicity}
A Markov chain with transition kernel $K$ is said to be ergodic if
\begin{align}
    \lim_{n\to\infty} \Vert K^n(q,\cdot) - \Pi(\cdot) \Vert_{\mathrm{TV}} = 0,
\end{align}
where for $A \in \mathfrak{B}(\R^m)$, $K^n(q, A) = \mathrm{Pr}(q^n \in A \vert q^0 =q)$ is the $n$-step Markov transition probability.
\end{definition}
It will turn out that one can establish ergodicity of a Markov chain if one can establish three separate properties: irreducibility, aperiodicity, and stationarity. We now define each of these concepts.
\begin{definition}
A Markov chain transition kernel $K$ is said to be $\Pi$-irreducible if for every set $A\in\mathfrak{B}(\R^m)$ with $\Pi(A) > 0$ there exists $n\in\mathbb{N}$ for which $K^n(q, A) > 0$ for all $q\in \R^m$.
\end{definition}
\begin{definition}
Given a Markov chain transition kernel $K : \R^m\times \mathfrak{B}(\R^m)\to [0, 1]$, a set $C\subseteq\R^m$ is called {\it small} if there exists a $n\in\mathbb{N}$, a $\delta > 0$, and a probability measure $\nu : \mathfrak{B}(\R^m)\to[0, 1]$ such that for any $q\in C$ and $A\in \mathfrak{B}(\R^m)$ we have $K^n(q, A) \geq \delta \cdot \nu(A)$.
\end{definition}
We may therefore say that the set $C$ is $(n, \delta, \nu)$-small.
\begin{definition}
A Markov chain transition kernel $K$ is called aperiodic if there exists a small set $C$ for which the greatest common divisor of the set
\begin{align}
    \set{n' \in \mathbb{N} : \exists~ \delta > 0 ~\text{and probability measure}~ \nu ~\text{such that}~ C ~\text{is}~ (n',\delta, \nu)\text{-small} }
\end{align}
is one.
\end{definition}
\begin{definition}\label{geometric-ergodicity:def:stationarity}
A Markov chain with transition kernel $K$ is said to be stationary for the probability measure $\Pi$ if for any Borel set $A\in\mathfrak{B}(\R^m)$ we have $\underset{q\sim \Pi}{\mathbb{E}} K(q, A) = \Pi(A)$.
\end{definition}
\begin{definition}\label{geometric-ergodicity:def:geometric-ergodicity}
A Markov chain with transition kernel $K$ is called geometrically ergodic if there exists $\rho\in (0, 1)$ and function $V : \R^m\to\R_+$ such that
\begin{align}
    \Vert K^n(q,\cdot) - \Pi(\cdot) \Vert_{\mathrm{TV}} \leq V(q) \rho^n.
\end{align}
\end{definition}
It is clear from the definitions that geometric ergodicity is a stronger form of convergence than mere ergodicity. In the latter case, however, simple conditions under which a Markov chain is ergodic can be provided.
\begin{theorem}[\citet{10.1214/aos/1176325750}]\label{geometric-ergodicity:thm:irreducible-aperiodic-stationary}
Suppose that $K$ is a Markov chain transition kernel that is $\Pi$-irreducible, aperiodic, and for which $\Pi$ is the stationary probability measure. Then $K$ produces an ergodic Markov chain.
\end{theorem}
However, in addition to giving an upper bound on the mixing of the Markov chain into the target probability measure, geometric ergodicity is also useful for establishing a central limit theorem for expectations computed from the sequence $(q^1, q^2,\ldots)$. We make this precise as follows.
\begin{theorem}[\citet{meyn1993markov}]
Let $h : \R^m\to \R$ be a function for which $\underset{q\sim \Pi}{\mathbb{E}} \abs{h(q)}^{2+\gamma}<\infty$ for some $\gamma > 0$. Let $K$ be a Markov chain transition kernel that is aperiodic, $\Pi$-irreducible, and for which $\Pi$ is the unique stationary distribution. Assume further that $K$ converges geometrically to $\Pi$. Define $\mathcal{S}_K[h] = n^{-1} \sum_{i=1}^n h(q^i)$ where $q^{i+1}\vert q^i \sim K(q^i, \cdot)$. Then, as $n\to\infty$,
\begin{align}
    \sqrt{n} \paren{\mathcal{S}_K[h] - \underset{q\sim \Pi}{\mathbb{E}}\left[ h(q)\right]} \overset{d}{\to} \mathrm{Normal}(0, \tau^2_h)
\end{align}
where $\tau^2_h$ is a constant, less than infinity, that depends on $h$ (and $K$). The quantity $\tau^2_h$ has a closed-form given by
\begin{align}
    \tau^2_h = \underset{q\sim \Pi}{\mathrm{Var}}\left(h(q)\right) + 2 \sum_{k=1}^\infty \underset{q\sim \Pi}{\mathrm{Cov}}\left(h(q), h(q^k)\right).
\end{align}
\end{theorem}
Establishing that a Markov chain has $\Pi$ as a stationary distribution is often easily achieved by showing that the chain satisfies detailed balance with respect to $\Pi$. As we require detailed balance when discussing the marginal transition kernels of RMHMC and LMC in \cref{geometric-ergodicity:subsec:inherited-geometric-ergodicity}, we define this notion now.
\begin{definition}\label{geometric-ergodicity:def:detailed-balance}
A Markov chain transition kernel is said to satisfy {\it detailed balance} with respect to the probability distribution $\Pi$ (equivalently, the Markov chain is called {\it reversible} with respect to $\Pi$) if for any sets $Q, Q'\in\mathfrak{B}(\R^m)$,
\begin{align}
    \int_Q K(q, Q')~\Pi(\mathrm{d}q) = \int_{Q'} K(q, Q) ~\Pi(\mathrm{d}q).
\end{align}
\end{definition}
\begin{remark}
When detailed balance holds for a Markov chain transition kernel $K$, it follows that $\Pi$ is the stationary probability measure (\cref{geometric-ergodicity:def:stationarity}) of the Markov transition kernel $K$.
\end{remark}

\subsubsection{Metropolis-Hastings Kernels and Generalized Langevin Algorithms}

In this section we review Metropolis-Hastings Markov chains. We begin by formally defining these objects before proceeding to a special case of Metropolis-Hastings method based on Langevin diffusion, which recalls the constructions considered in \citet{roy2021convergence}.
\begin{definition}[Metropolis-Hastings Algorithm]
Let $\pi:\R^m\to\R_+$ be a probability density and, for each $q\in\R^m$, let $\tilde{\pi}(\cdot \vert q):\R^m\to\R_+$ be a probability density depending on $q$. The Metropolis-Hastings transition kernel with proposal density $\tilde{\pi}(\cdot\vert q)$ is,
\begin{align}
\begin{split}
    K(q, A) &= \int_{A} \min\set{1, \frac{\pi(\tilde{q}) \tilde{\pi}(q\vert \tilde{q})}{\pi(q) \tilde{\pi}(\tilde{q}\vert q)}} \tilde{\pi}(\tilde{q}\vert q)~\mathrm{d}\tilde{q} \\
    &\qquad +~ \paren{1 - \int_{\R^m} \min\set{1, \frac{\pi(\tilde{q}) \tilde{\pi}(q\vert \tilde{q})}{\pi(q) \tilde{\pi}(\tilde{q}\vert q)}} \tilde{\pi}(\tilde{q}\vert q)~\mathrm{d}\tilde{q}} \mathrm{1}\set{q\in A},
\end{split}
\end{align}
for $q\in \R^m$ and $A\in\mathfrak{B}(\R^m)$.
\end{definition}
\begin{definition}[Generalized Metropolis-Adjusted Langevin Algorithm]\label{geometric-ergodicity:def:generalized-langevin}
Fix $\epsilon\in \R\setminus \set{0}$. Let $c_\epsilon : \R^m\to\R^m$ and let $\mathbf{A}:\R^m\to\mathrm{PD}(m)$. The generalized Metropolis-adjusted Langevin algorithm is an instance of the Metropolis-Hastings algorithm with proposal density
\begin{align}
    \tilde{\pi}(\tilde{q}\vert q) = \mathrm{Normal}(\tilde{q}; c_\epsilon(q), \epsilon^2 \mathbf{A}(q)).
\end{align}
We denote the Markov chain transition kernel of the generalized Metropolis-adjusted Langevin algorithm with step-size $\epsilon$ by $J_\epsilon : \R^m\times \mathfrak{B}(\R^m)\to [0,1]$.
\end{definition}
\begin{definition}[Riemannian Manifold Metropolis-Adjusted Langevin Algorithm]\label{geometric-ergodicity:def:riemannian-manifold-langevin}
In the special case where
\begin{align}
    \label{geometric-ergodicity:eq:euler-maruyama-manifold-langevin} c_\epsilon(q) &= \frac{\epsilon^2}{2} \mathbf{A}(q) \nabla\log \pi(q) + \frac{\epsilon^2}{2} \Gamma(q) \\
    \Gamma_i(q) &= \sum_{j=1}^m \frac{\partial}{\partial q_j} \mathbf{A}_{ij}(q)
\end{align}
we call the resulting method the Riemannian manifold Metropolis-adjusted Langevin algorithm (MMALA).
\end{definition}
\begin{remark}
The form of the proposal in \cref{geometric-ergodicity:eq:euler-maruyama-manifold-langevin} is based on an Euler-Maruyama discretization of a Langevin diffusion on a manifold with Riemannian metric $\mathbf{G}(q) = \mathbf{A}^{-1}(q)$; see \citet{XIFARA201414} for details. In summary, \cref{geometric-ergodicity:eq:euler-maruyama-manifold-langevin} is the drift component of the Euler-Maruyama discretization (with step-size $\epsilon^2$) applied to the following stochastic differential equation:
\begin{align}
    \label{geometric-ergodicity:eq:manifold-langevin} \mathrm{d}X_t = \frac{1}{2}\mathbf{G}^{-1}(X_t) \nabla \log \pi(X_t) ~\mathrm{d}t + \Gamma(X_t)~\mathrm{d}t + \mathbf{G}^{-1/2}(X_t)~\mathrm{d}B_t,
\end{align}
where $B_t$ is Euclidean Brownian motion at time $t\in\R$. When $\mathbf{G}(q)$ is the Fisher information matrix, the term $\mathbf{G}^{-1}(X_t) \nabla \log \pi(X_t)$ can be identified as the {\it natural gradient} of the function $\log \pi$ under the geometry generated by the Fisher metric \citep{amari2000methods} at position $X_t$. On the other hand, the term $\Gamma(X_t)~\mathrm{d}t + \mathbf{G}^{-1/2}(X_t)~\mathrm{d}B_t$ corresponds to {\it manifold Brownian motion}, since its infinitesimal generator is the Laplace-Beltrami operator on the manifold \citep{hsu2002stochastic}. Therefore, as the stochastic differential equation in \cref{geometric-ergodicity:eq:manifold-langevin} consists of a term comprising the manifold gradient of a log-density and another term comprising manifold Brownian motion, it is called the {\it Riemannian Langevin equation} in correspondence with the Euclidean case.
\end{remark}
\begin{definition}[Simplified Riemannian Manifold Metropolis-Adjusted Langevin Algorithm]\label{geometric-ergodicity:def:simplified-manifold-langevin}
Computing the function $\Gamma$ in \cref{geometric-ergodicity:def:riemannian-manifold-langevin} may be inconvenient to evaluate. It can be ignored while still yielding a Markov chain transition kernel, giving the special case,
\begin{align}
    c_\epsilon(q) = \frac{\epsilon^2}{2} \mathbf{A}(q) \nabla\log \pi(q).
\end{align}
The resulting method is called the simplified Riemannian manifold Metropolis-adjusted Langevin algorithm, which was considered by \citet{rmhmc}.
\end{definition}
We adopt the abbreviation SMALA to refer to the simplified Metropolis-adjusted Langevin algorithm.
\begin{remark}
There are three common choices of the function $\mathbf{A}$. The first is that $\mathbf{A}$ is a constant function, in which case we simply write $\mathbf{A}\in\mathrm{PD}(m)$. A second option is that $\mathbf{A}$ is chosen as the {\it inverse} of the sum of the Fisher information matrix and the negative Hessian of the log-prior, which can capture second-order geometry of both the likelihood and the prior; the use of the inverse of the sum of the Fisher information and the negative Hessian of the log-prior as a preconditioner is the approach advocated by \citet{rmhmc}. A third option is that $\mathbf{A}$ is the inverse of the SoftAbs metric, which is a smooth transformation of the Hessian of the log-density of the target distribution; for details see \citet{softabs}.
\end{remark}

\subsubsection{Geometric Methods of Bayesian Inference}

\begin{definition}[Involutive Monte Carlo \citep{DBLP:conf/icml/NeklyudovWEV20}]\label{def:involutive-monte-carlo}
Let $\Phi : \R^m\to\R^m$ be a smooth involution and let $\pi : \R^m\to\R_+$ be a probability density on $\R^m$. The Markov chain transition kernel of involutive Monte Carlo with target density $\pi$ is
\begin{align}
\begin{split}
    K(q, A) &= \min\set{1, \frac{\pi(\Phi(q))}{\pi(q)} \abs{\mathrm{det}(\nabla \Phi(q))}} \mathbf{1}\set{\Phi(q) \in A} \\
    &\qquad +~ \paren{1 - \min\set{1, \frac{\pi(\Phi(q))}{\pi(q)} \abs{\mathrm{det}(\nabla \Phi(q))}}} \mathbf{1}\set{q\in A},
\end{split}
\end{align}
for $q\in\R^m$ and $A\in\mathfrak{B}(\R^m)$.
\end{definition}
It is easily verified that the transition kernel of involutive Monte Carlo satisfies detailed balance with respect to the density $\pi$. Involutive Monte Carlo gives rise to two special transition kernels corresponding to RMHMC and LMC.
\begin{definition}\label{def:rmhmc-transition-kernel}
Let $\hat{\Phi}_\epsilon$ be as in \cref{geometric-ergodicity:def:riemannian-leapfrog} and let $H$ be as in \cref{geometric-ergodicity:def:riemannian-hamiltonian}. 
The involution of Riemannian manifold Hamiltonian Monte Carlo with step-size $\epsilon$ and $k$ integration steps is,
\begin{align}
    \label{eq:rmhmc-involution}\Phi = \mathbf{F} \circ \underbrace{\hat{\Phi}_\epsilon \circ\cdots\circ \hat{\Phi}_\epsilon}_{k~\mathrm{times}},
\end{align}
where $\mathbf{F}(q, p) = (q, -p)$. The target density of Riemannian manifold Hamiltonian Monte Carlo is $\pi(q, p) \propto \exp(-H(q, p))$. The transition kernel of involutive Monte Carlo (\cref{def:involutive-monte-carlo}) with involution $\Phi$ given in \cref{eq:rmhmc-involution} is called the transition kernel of Riemannian manifold Hamiltonian Monte Carlo with step-size $\epsilon$ and $k$ integration steps.
\end{definition}
\begin{remark}
It can be shown that if $\Phi$ is an invertible function and if $\mathbf{F}\circ\Phi$ is an involution that $\mathbf{F}\circ\Phi\circ\cdots\cdot\Phi$ is also an involution. This explains why \cref{eq:rmhmc-involution} is also an involution. For further details, see \cref{prop:involution-composition}.
\end{remark}
% \begin{remark}
In the special case where $\mathbf{G}$ is a constant function of $q$, the resulting Markov chain is called Euclidean Hamiltonian Monte Carlo.
\begin{definition}
Let $\check{\Phi}_\epsilon$ be as in \cref{geometric-ergodicity:def:euclidean-leapfrog} and let $H$ be as in \cref{geometric-ergodicity:def:euclidean-hamiltonian}. 
The involution of Euclidean Hamiltonian Monte Carlo (EHMC) with step-size $\epsilon$ and $k$ integration steps is,
\begin{align}
    \label{eq:ehmc-involution}\Phi = \mathbf{F} \circ \underbrace{\check{\Phi}_\epsilon \circ\cdots\circ \check{\Phi}_\epsilon}_{k~\mathrm{times}},
\end{align}
where $\mathbf{F}(q, p) = (q, -p)$. The target density of Euclidean Hamiltonian Monte Carlo is $\pi(q, p) \propto \exp(-H(q, p))$. The transition kernel of involutive Monte Carlo (\cref{def:involutive-monte-carlo}) with involution $\Phi$ given in \cref{eq:ehmc-involution} is called the transition kernel of Euclidean Hamiltonian Monte Carlo with step-size $\epsilon$ and $k$ integration steps.
\end{definition}

\begin{definition}
Let $\tilde{\Phi}_\epsilon$ be as in \cref{geometric-ergodicity:def:lagrangian-leapfrog} and let $H$ be as in \cref{geometric-ergodicity:def:riemannian-hamiltonian}. The involution of Lagrangian Monte Carlo with step-size $\epsilon$ and $k$ integration steps is,
\begin{align}
    \Phi = \mathbf{F} \circ \underbrace{\tilde{\Phi}_\epsilon \circ\cdots\circ \tilde{\Phi}_\epsilon}_{k~\mathrm{times}},
\end{align}
with $\mathbf{F}(q, p) = (q, -p)$. The target density of Lagrangian Hamiltonian Monte Carlo is $\pi(q, p) \propto \exp(-H(q, p))$. We call this the transition kernel of Lagrangian Monte Carlo with step-size $\epsilon$ and $k$ integration steps.
\end{definition}
\begin{remark}
Unlike the transition kernel employed in Riemannian manifold Hamiltonian Monte Carlo, the Jacobian determinant of the involution appearing in Lagrangian Monte Carlo must be computed through $k$ applications of \cref{geometric-ergodicity:eq:lagrangian-jacobian-determinant}, where, as before, $k$ is the number of integration steps.
\end{remark}
In the special case where $\mathbf{G}$ is a constant function of $q$, Riemannian manifold Hamiltonian Monte Carlo and Lagrangian Monte Carlo produces a transition kernel that is equivalent to Euclidean Hamiltonian Monte Carlo with constant mass matrix $\mathbf{G}$.
\begin{proposition}
The transition kernels of RMHMC (or LMC) with step-size $\epsilon$ and $k$ integration steps when applied to the Euclidean Hamiltonian in \cref{geometric-ergodicity:def:euclidean-hamiltonian} (in the sense that we take $\mathbf{G}(q) \equiv \mathbf{G}$ for every $q\in\R^m$) are both equivalent to the EHMC transition kernel.
\end{proposition}
\begin{proof}
It follows from \cref{geometric-ergodicity:prop:generalized-langrangian-euclidean-equivalent} that the proposals produced by RMHMC and LMC are identical to EHMC. It remains to be verified that the acceptance decisions are identical, too. This can be verified by observing that the Riemannian Hamiltonian in \cref{geometric-ergodicity:eq:riemannian-hamiltonian} is the Euclidean Hamiltonian in \cref{geometric-ergodicity:eq:euclidean-hamiltonian} up to an additive constant (since $\mathbf{G}$ does not depend on position). Hence, the Euclidean and Riemannian densities are proportional to one another, and the acceptance probabilities will be identical.
\end{proof}
We now elaborate on the connection between EHMC {\it with a single-step} and the Metropolis-adjusted Langevin algorithm: in particular, these two methods can be constructed to be exactly equivalent in $q$-space, producing equal proposals and identical acceptance probabilities.
\begin{proposition}[\citet{1206.1901}]\label{geometric-ergodicity:prop:mala-hmc-equivalence}
Let $\mathbf{A}\in\R^{m\times m}$ be a fixed positive definite matrix. Consider EHMC with the Hamiltonian $H(q,p)=-\log\pi(q) + \frac{1}{2} p^\top\mathbf{A}^{-1}p$ and MALA with proposal distribution $\tilde{q}\vert q \sim\mathrm{Normal}\paren{q+\frac{\epsilon^2}{2}\mathbf{A}^{-1}\nabla \log\pi(q), \epsilon^2\mathbf{A}^{-1}}$. The marginal transition kernel of EHMC with step-size $\epsilon$ and a single integration step can be constructed to be exactly equivalent to the transition kernel of MALA.
\end{proposition}
A proof is provided in \cref{geometric-ergodicity:app:proof-mala-hmc-equivalence}.
The analysis of mixture transition kernels will be central to our analysis. In standard implementations of Riemannian manifold Hamiltonian and Lagrangian Monte Carlo, it is typical to randomize the number of integration steps. This is done to avert any Markov chain pathologies (such as irreducibility failures) that may result from using a fixed number of integration steps. We therefore consider Markov chain transition kernels which are mixtures.
\begin{definition}\label{geometric-ergodicity:def:unmodified-kernel}
Let $K_{\epsilon, k} : (\R^m\times\R^m)\times \mathfrak{B}(\R^m\times\R^m) \to\R_+$ be the Markov chain transition kernel of RMHMC (or LMC) with step-size $\epsilon$ and $k$ integration steps. Let $\tilde{K}_{\epsilon, k}$ be the marginal transition kernel (defined in \cref{geometric-ergodicity:lem:marginal-transition-kernel}) of $K_{\epsilon, k}$. Define the {\it mixture transition kernel of RMHMC (or LMC)} to be
\begin{align}
    \tilde{K}_\epsilon = \sum_{k=1}^\infty \alpha_k \tilde{K}_{\epsilon, k},
\end{align}
where $(\alpha_1,\alpha_2,\ldots)$ is a probability vector.
\end{definition}
The fact that EHMC, RMHMC, and LMC all satisfy detailed balance in $(q, p)$-space causes us to examine the progression of $q$-states alone in between Gibbs resampling steps of the momentum. This leads to {\it marginal} Markov chain transition kernels of the following form:
\begin{lemma}\label{geometric-ergodicity:lem:marginal-transition-kernel}
  Let $K : (\R^{m}\times\R^m) \times \mathfrak{B}(\R^{m}\times\R^m)\to\R$ be a Markov chain transition kernel. Suppose that $K$ satisfies detailed balance with respect to the density $\pi(q, p)$ with $q$-marginal distribution $\pi(q)$ and conditional density $\pi(p\vert q)$; i.e. $\pi(q, p) = \pi(q) \pi(p\vert q)$. 
  Consider the marginal chain constructed as follows. Given $q^n=q$, sample $p^{n}\sim \pi(p\vert q)$, sample $(q^{n+1}, p^{n+1}) \sim K((q^n, p^n), \cdot)$ and discard both momenta.
  The transition kernel of the marginal chain satisfies
  \begin{align}
      \tilde{K}(q, Q) = \int_{\R^m} K((q, p), (Q, \R^m)) ~\pi(p\vert q) ~\mathrm{d}p
  \end{align}
  where $Q\in\mathfrak{B}(\R^m)$.
\end{lemma}
A proof is given in \cref{geometric-ergodicity:app:proof-marginal-transition-kernel}
\begin{proposition}\label{geometric-ergodicity:prop:marginal-reversible}
Let $\tilde{K}$ be the marginal transition kernel described in \cref{geometric-ergodicity:lem:marginal-transition-kernel}. The marginal chain satisfies detailed balance with respect to the distribution whose density is $\pi(q)$.
\end{proposition}
A proof is provided in \cref{geometric-ergodicity:app:proof-marginal-reversible}.

\subsubsection{Ergodicity Theorems}

The geometric ergodicity of MALA, MMALA, and SMALA have been examined in the Markov chain literature. In \cref{geometric-ergodicity:subsec:inherited-geometric-ergodicity} we will see how the geometric ergodicity of a single component of a mixture Markov chain transition kernel implies the geometric ergodicity of the mixture itself. With this destination in mind, we now recall two key results in this direction.
\begin{theorem}[\citet{bj/1178291835}]
Consider the transition kernel $K$ of the Metropolis-adjusted Langevin algorithm with $\mathbf{A} = \mathrm{Id}$ and $c_\epsilon(q) = x + \frac{\epsilon^2}{2}\nabla \log\pi(q)$. Then $K$ is geometrically ergodic under the following conditions:
\begin{enumerate}
    \item We have
    \begin{align}
        \liminf_{\Vert q\Vert\to\infty} \paren{\Vert q\Vert - \Vert q + \frac{\epsilon^2}{2} \nabla \log\pi(q)\Vert} > 0.
    \end{align}
    \item We have
    \begin{align}
        \lim_{\Vert q\Vert \to\infty} \int_{(A(q)\cup I(q)) \cap (A(q) \cap I(q))} \tilde{\pi}(\tilde{q}\vert q) ~\mathrm{d}\tilde{q} = 0,
    \end{align}
    where 
    \begin{align}
         A(q) &= \set{\tilde{q}\in\mathcal{X} : \pi(q) \tilde{\pi}(\tilde{q}\vert q) \leq \pi(y)\tilde{\pi}(q\vert \tilde{q})} \\
         I(q) &= \set{\tilde{q}\in\mathcal{X} : \Vert \tilde{q}\Vert \leq \Vert q\Vert}.
    \end{align}
\end{enumerate}
\end{theorem}

\begin{theorem}[\citet{roy2021convergence}]\label{geometric-ergodicity:thm:generalized-langevin-geometric-ergodicity}
Consider the transition kernel $K$ of the generalized Metropolis-adjusted Langevin algorithm. Then $K$ is geometrically ergodic under the following conditions:
\begin{enumerate}
    \item There exist matrices $\mathbf{A}_l \in \mathrm{PD}(m)$ and $\mathbf{A}_u \in \mathrm{PD}(m)$ such that $\mathbf{A}_l \leq \mathbf{A}(q) \leq \mathbf{A}_u$.
    \item When $A\subset\R^m$ is bounded, the function $c_\epsilon:\R^m\to\R^m$ is bounded on $A$.
    \item There is a quantity
    \begin{align}
        C = \limsup_{\Vert q\Vert\to\infty} \int_{R(q)} \paren{1 - \tilde{\pi}(\tilde{q}\vert q)\min\set{1, \frac{\pi(\tilde{q})\tilde{\pi}(q\vert \tilde{q})}{\pi(q)\tilde{\pi}(\tilde{q}\vert q)}}} ~\mathrm{d}y
    \end{align}
    which is strictly less than one, where 
    \begin{align}
        R(q) = \set{ \tilde{q} \in \R^m : \pi(q)\tilde{\pi}(\tilde{q}\vert x) > \pi(\tilde{q})\tilde{\pi}(q\vert \tilde{q})}.
    \end{align}
    \item There exists $s > 0$ such that,
    \begin{align}
        \liminf_{\Vert q\Vert\to\infty} \paren{\Vert \mathbf{A}^{-1/2}_u(q) x\Vert - \Vert \mathbf{A}^{-1/2}_u(q) c_\epsilon(q)\Vert} > \frac{\log(D(s)) - \log (1-C)}{s},
    \end{align}
    where,
    \begin{align}
        D(s) = \epsilon^{-m/2} \paren{\frac{\pi}{2}}^{m/2-1} \paren{\frac{\mathrm{det}(\mathbf{A}_u)}{\mathrm{det}(\mathbf{A}_l)}}^{1/2} \exp(\epsilon s^2 / 2) \int_0^\infty \exp\paren{ - (r - \epsilon s)^2 / 2\epsilon} r^{m-1} ~\mathrm{d}r
    \end{align}
\end{enumerate}
\end{theorem}
In \cref{geometric-ergodicity:subsec:inherited-geometric-ergodicity} we will examine mixture transition kernels and their geometric ergodicity. In the case that one transition kernel is geometrically ergodic and all transition kernels are reversible, the mixture transition kernel is geometrically ergodic, too, from the following result.
\begin{theorem}[\citet{10.2307/43304674}]\label{geometric-ergodicity:thm:mixture-geometric-ergodicity}
Let $(\alpha_1,\alpha_2, \ldots)$ be a sequence satisfying $\alpha_i\geq 0$ and $\sum_{i=1}^\infty \alpha_i = 1$. Let $\tilde{K} = \sum_{i=1}^\infty \alpha_i K_i$ be a mixture of reversible transition kernels with invariant distribution $\Pi$. Suppose that $K_1$ satisfies the properties: (i) $\Pi$ is the {\it unique} invariant distribution, and (ii) $K_1$ is geometrically ergodic. Then $\tilde{K}$ is geometrically ergodic.
\end{theorem}
We term the derived geometric ergodicity of $\tilde{K}$ from $K_1$ as ``inherited'' geometric ergodicity.

\section{Inherited Geometric Ergodicity}\label{geometric-ergodicity:subsec:inherited-geometric-ergodicity}

In this section we discuss the notion of ``inherited'' geometric ergodicity, wherein we modify RMHMC and LMC to be geometrically ergodic when MMALA is. 
In the Euclidean case, by \cref{geometric-ergodicity:prop:mala-hmc-equivalence}, in cases where the Metropolis-adjusted Langevin algorithm is geometrically ergodic, so is HMC by invoking \cref{geometric-ergodicity:thm:mixture-geometric-ergodicity}. The situation in the geometric setting is more complicated, since there is no analogue of \cref{geometric-ergodicity:prop:mala-hmc-equivalence} to show that the marginal transition kernel of a single-step of RMHMC or LMC is exactly equivalent to the Markov chain transition kernel of the generalized Metropolis-adjusted Langevin algorithm for a suitable choice of mean function. This leads us to propose the following Markov chain transition kernel.
\begin{definition}
Let $K_{\epsilon, k} : (\R^m\times\R^m)\times \mathfrak{B}(\R^m\times\R^m) \to\R_+$ be the Markov chain transition kernel of RMHMC (or LMC) with step-size $\epsilon$ and $k$ integration steps. Let $\tilde{K}_{\epsilon, k}$ be the marginal transition kernel (defined in \cref{geometric-ergodicity:lem:marginal-transition-kernel}). Let $(\alpha_1,\alpha_2, \ldots)$ be a sequence satisfying $\alpha_i\geq 0$ and $\sum_{i=1}^\infty \alpha_i = 1$. Let $J_\epsilon$ be the Markov chain transition kernel of the generalized Metropolis-adjusted Langevin algorithm (defined in \cref{geometric-ergodicity:def:generalized-langevin}). The {\it Langevin  mixture transition kernel of RMHMC (or LMC)}, which we abbreviate by LMRMHMC (or LMLMC), is defined by
\begin{align}
    \label{geometric-ergodicity:eq:modified-marginal-transition} \tilde{K}_{\epsilon} = \alpha_1 J_\epsilon + \sum_{k=2}^\infty  \alpha_k \tilde{K}_{\epsilon, k}.
\end{align}
We call $\alpha_1$ the MMALA mixture weight.
\end{definition}
The modified transition kernel simply replaces a single-step of RMHMC (or LMC) by the transition kernel of the generalized Metropolis-adjusted Langevin algorithm. In order to apply \cref{geometric-ergodicity:thm:mixture-geometric-ergodicity}, it is necessary to verify that the {\it marginal} transition kernels of RMHMC and LMC are reversible. Fortunately, this is readily shown as follows.
\begin{corollary}\label{geometric-ergodicity:cor:rmhmc-lmc-marginal-reversible}
For fixed $\epsilon\in \R$ and $k\in\mathbb{N}$, the RMHMC and LMC transition kernels with step-size $\epsilon$ and $k$ integration steps satisfy detailed balance in $(q, p)$-space, it follows that their marginal chains are reversible with respect to the distribution $\Pi$ by \cref{geometric-ergodicity:prop:marginal-reversible}.
\end{corollary}
\begin{proposition}
Suppose $\alpha_1 > 0$. Under the conditions of \cref{geometric-ergodicity:thm:generalized-langevin-geometric-ergodicity}, the Markov chain transition kernel of LMRMHMC (or LMLMC) in \cref{geometric-ergodicity:eq:modified-marginal-transition} is geometrically ergodic.
\end{proposition}
\begin{proof}
This follows as an immediate consequence of \cref{geometric-ergodicity:thm:mixture-geometric-ergodicity} using the fact that the marginal transition kernels or RMHMC (or LMC) are reversible by \cref{geometric-ergodicity:cor:rmhmc-lmc-marginal-reversible}.
\end{proof}
\begin{remark}
As a practical matter, we choose the mixture probabilities $(\alpha_1,\alpha_2,\ldots)$ in the following way. For a particular target distribution, we will choose a maximal number of integration steps $k_\mathrm{max}$ for which $\alpha_k = 0$ for $k > k_\mathrm{max}$. Then, given a particular selection of $\alpha_1 \in (0, 1]$, we split the remaining probability mass equally for each $k\in \set{2,\ldots, k_{\mathrm{max}}}$; that is, the fraction can be expressed as $\alpha_k = \frac{1 - \alpha_1}{k_{\mathrm{max}} - 1}$ for $k=2,\ldots, k_\mathrm{max}$. In our experiments, we will consider variable choices for the mixing parameter $\alpha_1$. The special case $\alpha_1 = 0$, will correspond by convention to the unmodified RMHMC and LMC transition kernels with a single-step computed using the prescribed involution, as described in \cref{geometric-ergodicity:def:unmodified-kernel}.
\end{remark}
Mixing with the MMALA transition kernel also immediately establishes irreducibility, aperiodicity, and the smallness of all compact sets, as the following result reveals.
\begin{lemma}\label{geometric-ergodicity:lem:irreducible-aperiodic-small}
Let $\pi(q) \propto \exp(\mathcal{L}(q))$ be continuous and bounded on compact sets and denote by $\Pi$ the probability measure with density $\pi$. Suppose $\alpha_1 > 0$. The marginal transition kernel of LMRMHMC (or LMLMC) is $\Pi$-irreducible, aperiodic, and all non-negligible compact sets are small.
\end{lemma}
\begin{proof}
With probability $\alpha_1$, the Markov chain transitions according to a Metropolis-Hastings accept-reject decision with a normal proposal distribution. Hence, for any set $A\subset \mathfrak{B}(\R^m)$ for which $\int_A \pi(q)~\mathrm{d}q > 0$, we have,
\begin{align}
    \tilde{K}_\epsilon(q, A) &\geq \alpha_1 J_\epsilon(q, A) \\
    &> 0
\end{align}
since a normal proposal distribution is non-vanishing everywhere on $\R^m$. The fact that $\tilde{K}_\epsilon$ is aperiodic and that all non-negligible compact sets are small follows as an immediate consequence of Lemma 1.2 from \citet{10.1214/aos/1033066201}.
\end{proof}
\begin{corollary}
From \cref{geometric-ergodicity:lem:irreducible-aperiodic-small}, we have that the modified Markov chain of RMHMC (or LMC) is $\Pi$-irreducible, aperiodic, and from \cref{geometric-ergodicity:prop:marginal-reversible} $\Pi$ is the stationary distribution. Hence, it follows from \cref{geometric-ergodicity:thm:irreducible-aperiodic-stationary} that the modified Markov chain of RMHMC (or LMC) produces an ergodic Markov chain.
\end{corollary}

\section{Experimentation}\label{geometric-ergodicity:sec:experiments}

We turn now to the investigation of the proposed modified variations of RMHMC and LMC. These examples are chosen to represent a wide class of posterior distributions. In computing the convergence of the Markov chain under the maximum mean discrepancy metric, we measure similarity between 10,000 i.i.d. samples of the target distribution, and an equal number of independent Markov chains; in this case, we compute the unbiased estimator of the maximum mean discrepancy over the course of the first one-hundred sampling steps. For the expected squared jump distance (ESJD) and the effective sample size (ESS) metrics, we consider a ten replicates of a long Markov chain consisting of 1,000,000 samples, except in the case of the Fitzhugh-Nagumo posterior where we sample only 100,000 times. In each experiment, we consider mixing LMC and RMHMC with the transition kernel described in \cref{geometric-ergodicity:def:riemannian-manifold-langevin}, except in the case of Neal's funnel distribution, wherein we consider a mixture with \cref{geometric-ergodicity:def:simplified-manifold-langevin}. We begin in \cref{geometric-ergodicity:subsec:measures-and-metrics} to describe measures and metrics by which we may assess the performance of the Markov chains corresponding to LMRHMC and LMLMC. 

Code for reproducing these experimental results may be found at \url{https://tinyurl.com/29kz7krz}.

\subsection{Measures and Metrics}\label{geometric-ergodicity:subsec:measures-and-metrics}

We now give details of the evaluation metrics by which we compare Markov chains. We first recall the method of maximum mean discrepancy due to \citet{gretton-kernel}.
\begin{definition}\label{geometric-ergodicity:def:mmd}
Let $k:\R^m\times\R^m\to\R$ be a positive definite function that is symmetric in its arguments. Let $\Pi$ and $\Pi'$ be two probability distributions on $\R^m$. The squared maximum mean discrepancy between $\Pi$ and $\Pi'$ is defined by,
\begin{align}
    \label{geometric-ergodicity:eq:mmd} \mathrm{MMD}^2(k, \Pi, \Pi') = \underset{q, q'\sim \Pi}{\mathbb{E}} k(q, q') + \underset{q, q'\sim \Pi'}{\mathbb{E}} k(q, q') - 2\underset{q\sim \Pi, q'\sim \Pi'}{\mathbb{E}} k(q, q').
\end{align}
Let $(q^1_*,\ldots, q^r_*)\sim\Pi$ and $(q^1,\ldots, q^s)\sim \Pi'$. An unbiased estimator of the squared maximum mean discrepancy is,
\begin{align}
\begin{split}
    \label{geometric-ergodicity:eq:mmd-unbiased} \mathrm{MMD}^2_\mathrm{u}(k, \set{q^i_*}_{i=1}^r, \set{q^i}_{i=1}^s) &= \frac{1}{r(r-1)} \sum_{i=1}^r \sum_{j\neq i}^r k(q^i_*, q^j_*) + \frac{1}{s(s-1)} \sum_{i=1}^s \sum_{j\neq i}^s k(q^i, q^j) \\
    &\qquad -~ \frac{2}{rs} \sum_{i=1}^r\sum_{j=1}^s k(q^i_*, q^j).
\end{split}
\end{align}
\end{definition}
In our evaluations we adopt a squared exponential positive definite kernel $k(q, q') = \exp(- \Vert q - q'\Vert^2 / 2h)$, where $h\in\R_+$ is a  parameter called the {\it kernel bandwidth}. In our experiments we set $h$ to be the median distance between i.i.d. samples from the target distribution.
Recall the definition of geometric ergodicity given in \cref{geometric-ergodicity:def:geometric-ergodicity}. The measure of convergence is the total variation norm, which we can discuss theoretically but cannot evaluate in a computational setting. Instead, we can measure convergence to the target distribution $\Pi$ as a function of $n$, the number of steps, by means of \cref{geometric-ergodicity:eq:mmd}, for which we can obtain an unbiased estimate via \cref{geometric-ergodicity:eq:mmd-unbiased} if we have samples from $K^n(q^0,\cdot)$ and samples from $\Pi$. Indeed, as in \cref{geometric-ergodicity:def:mmd}, let $(q^1_*,\ldots,q^r_*)\sim\Pi$ and, for a fixed initial position $q^0$ and number of steps $n$ let $(q^1,\ldots, q^s)\sim K^n(q^0, \cdot)$ and we can compute an unbiased estimate of the squared maximum mean discrepancy. In the latter case, $(q^1,\ldots, q^s)$ can be obtained by running $s$ independent Markov chains for $n$ steps; independent samples from the target distribution $\Pi$ may be available in certain benchmark cases. In our experiments, i.i.d. samples from the target distribution may be generated from the banana-shaped posterior, Neal's funnel distribution, the Fitzhugh-Nagumo posterior, and the multi-scale Student-$t$ distribution.

\begin{remark}
Recall the definition of geometric ergodicity given in \cref{geometric-ergodicity:def:geometric-ergodicity}. Taking logarithms reveals
\begin{align}
    \log \Vert K^n(q, \cdot) - \Pi(\cdot)\Vert_{\mathrm{TV}} \leq n \log \rho + \log V(q).
\end{align}
Therefore, one may claim to see {\it evidence of geometric ergodicity} if, as a function of $n$, there is a linear decrease in the total variation distance on a logarithmic scale. Of course, we cannot directly compute the total variation distance, but we may look for a similar negative linear trend when $\abs{\mathrm{MMD}_u^2}$ is plotted on a logarithmic scale.
\end{remark}

As an additional measure of ergodicity, we consider comparing Markov chain samples against i.i.d. samples via random projection onto one-dimensional sub-spaces. Let $(q^1_*,\ldots,q^r_*)\sim\Pi$ and let $q^n\sim K^n(q^0, \cdot)$ for $n=1,\ldots, s$. Let $u$ be a random unit vector. We compute the Kolmogorov-Smirnov (KS) statistic for the projections $(u^\top q^1_*, \ldots, u^\top q^r_*)$ and $(u^\top q^1, \ldots, u^\top q^s)$. Repeating this process for one-hundred randomly generated unit vectors yields a distribution over Kolmogorov-Smirnov statistics. The more tightly concentrated this distribution is near zero, the closer the distribution of Markov chain iterates $(q^1,\ldots, q^s)$ is to the collection of i.i.d. samples $(q^1_*,\ldots, q^r_*)$ from the target distribution.

We also consider the expected squared jump distance (ESJD) \citep{esjd}, which measures the dissimilarity between subsequent states of the Markov chain. Intuitively, a Markov chain that moves more efficiently through the sample space (higher ESJD) will exhibit smaller sample auto-correlation. Formally, the ESJD is $\underset{q\sim \Pi}{\mathbb{E}} \Vert q' - q\Vert^2$ where $q'\vert q \sim K(q, \cdot)$. Note that the choice of norm $\Vert\cdot\Vert$ is left to the practitioner and may be selected to capture geometric properties of the target distribution. This expectation is typically approximated as follows: let $q^k$ be the  state of the Markov chain at step $k$ and at step $k+1$ a candidate state is generated, denoted $q^{k+1}_\mathrm{prop}$; the proposal state is accepted with probability $\alpha_k$; the following empirical mean is then taken as our approximation to the ESJD
\begin{align}
    \mathrm{ESJD}(\set{(q^k, q^{k+1}_{\mathrm{prop}}, \alpha_k)}_{k=1}^n) = \frac{1}{n} \sum_{k=1}^n \alpha_k \Vert q^{k+1}_{\mathrm{prop}} - q^k \Vert^2.
\end{align}
In our experiments in \cref{subsec:student-t}, we observe that the ESJD can be misleading as a measure. Therefore, we also introduce the {\it median} squared jump distance (MSJD) which we define as
\begin{align}
     \mathrm{MSJD}(\set{(q^k, q^{k+1}_{\mathrm{prop}}, \alpha_k)}_{k=1}^n) = \mathrm{Median}\paren{\set{\alpha_k \Vert q^{k+1}_{\mathrm{prop}} - q^k \Vert^2}_{k=1}^n}.
\end{align}
As the median, we expect the MSJD to exhibit less sensitivity to outliers than the ESJD.

We additionally consider the effective sample size (ESS) as a metric for our Markov chain procedures. We compute ESS using the technique of \citet{arviz_2019}, who describe the method succinctly as follows. For $i=1,\ldots, p$, let $(q^{i,1}, q^{i,2},\ldots,q^{i,n})$ be a sequence of $\R^m$-valued parameters. The integer $p$ is called the number of {\it chains}. The effective sample size of the $j$-th parameter is computed according to,
\begin{align}
    \mathrm{ESS}\paren{\set{q^{i, 1}_j, \ldots, q^{i, n}_j}_{i=1}^{p}} &= \frac{pn}{\hat{\tau}} \\
    \hat{\tau} &= \paren{2\sum_{k=1}^r \hat{\rho}_{2k} + \hat{\rho}_{2k+1}} - 1,
\end{align}
where $\hat{\rho}_k$ is an estimate, based on all of the $p$ sequences, of the autocorrelation of the $j$-th parameter with a $k$-step lag (for details see \citet{vehtari2021}) and $r$ is the smallest integer for which $\hat{\rho}_{2(r+1)} + \hat{\rho}_{2(r+1)+1} \leq 0$. In our experiments, we set $p=2$ by taking a single long Markov chain and splitting it in half at the middle. As a practical matter, we will also report the minimum, over all parameters of the posterior, ESS {\it per second} in order to represent the computational efficiency of the method.

\subsection{Banana-Shaped Distribution}

\begin{figure}[t!]
  \begin{subfigure}[t]{0.3\textwidth}
    \centering
    \includegraphics[width=\textwidth]{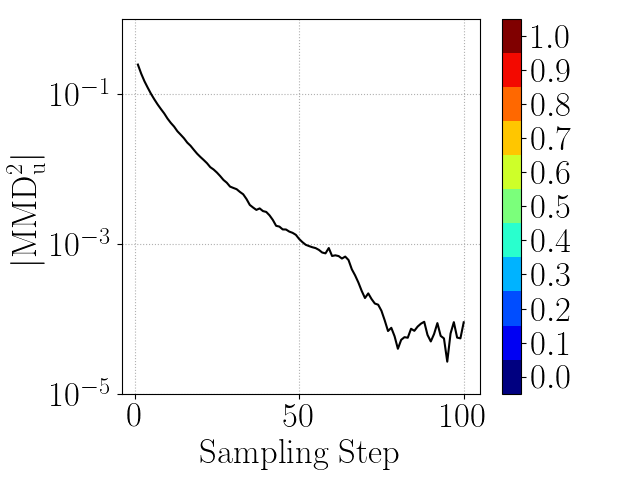}
    \caption{EHMC}
  \end{subfigure}
  ~
  \begin{subfigure}[t]{0.3\textwidth}
    \centering
    \includegraphics[width=\textwidth]{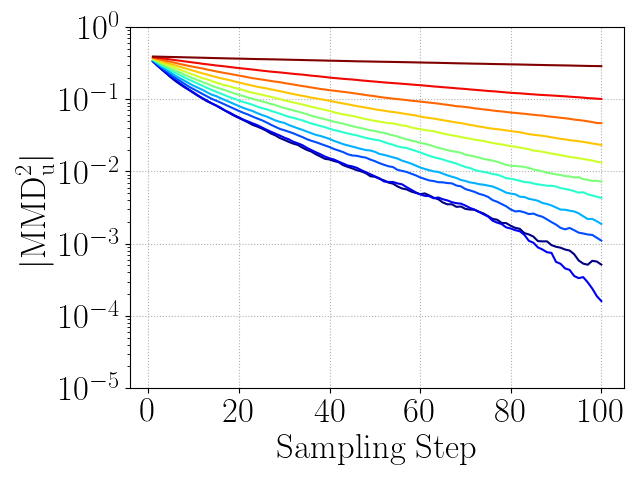}
    \caption{LMRMHMC}
  \end{subfigure}
  ~
  \begin{subfigure}[t]{0.3\textwidth}
    \centering
    \includegraphics[width=\textwidth]{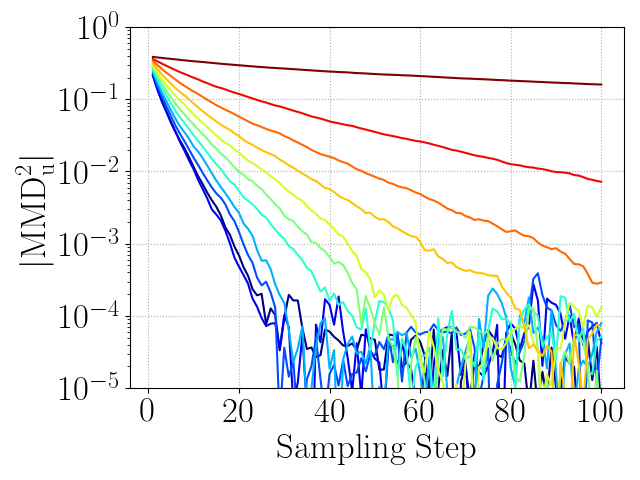}
    \caption{LMLMC}
  \end{subfigure}
  \caption{We examine the ergodicity of the Markov chains as a function of the sampling step in the banana-shaped posterior distribution. We observe that the use of a smaller step-size degrades the ergodicity of the RMHMC, particularly when compared against the LMC algorithm. In both RMHMC and LMC, we observe that aggressively mixing with MMALA causes the Markov chain to mix more slowly; on the other hand, a relatively small mixing probability leads to ergodicity that is nearly indistinguishable from the unmodified implementations of RMHMC and LMC.}
  \label{geometric-ergodicity:fig:banana-ergodicity}
\end{figure}

\begin{figure}[t!]
  \begin{subfigure}[t]{0.3\textwidth}
    \centering
    \includegraphics[width=\textwidth]{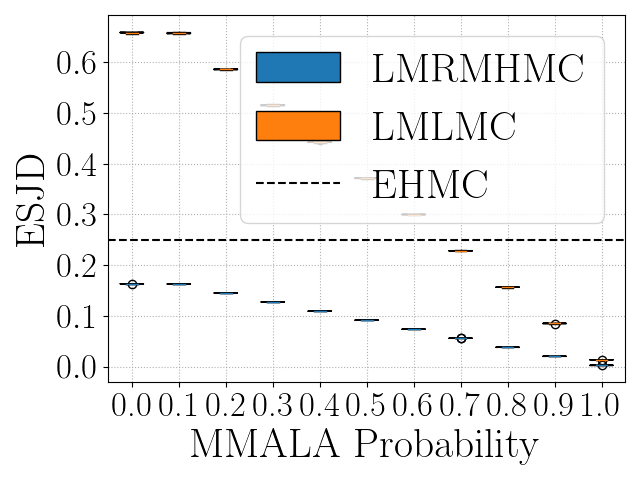}
    \caption{ESJD}
  \end{subfigure}
  ~
  \begin{subfigure}[t]{0.3\textwidth}
    \centering
    \includegraphics[width=\textwidth]{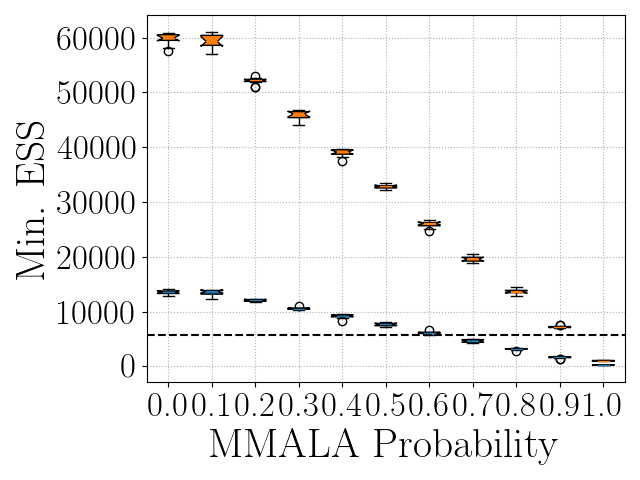}
    \caption{Min. ESS}
  \end{subfigure}
  ~
  \begin{subfigure}[t]{0.3\textwidth}
    \centering
    \includegraphics[width=\textwidth]{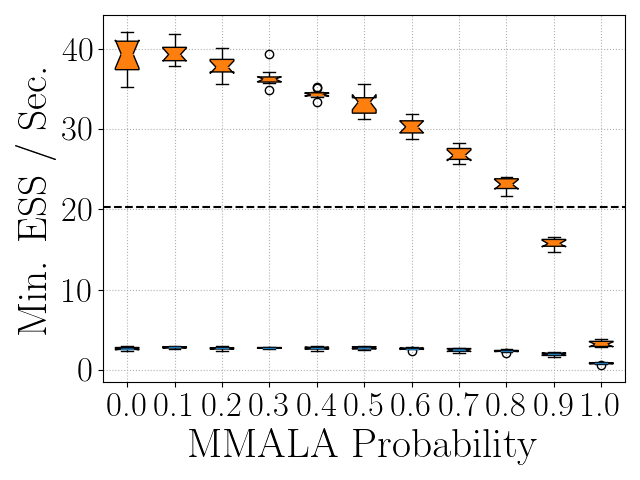}
    \caption{Min. ESS / Sec.}
  \end{subfigure}
  \caption{We show the ESJD, the minimum ESS, and the time-normalized minimum ESS for inference in the banana-shaped posterior distribution. We observe that for both RMHMC and LMC, aggressively mixing with MMALA causes the ESJD distance to decrease; this correspondingly produces a decrease in the minimum ESS. For LMC, the time-normalized minimum ESS is a decreasing function of the mixing probability.}
  \label{geometric-ergodicity:fig:banana-metrics}
\end{figure}

Our first example considers the following generative model, which represents an example of non-identifiable parameters.
\begin{align}
    \theta_1,\theta_2 &\overset{\mathrm{i.i.d.}}{\sim} \mathrm{Normal}(0, \sigma_\theta^2) \\
  y_i \vert\theta_1,\theta_2 &\overset{\mathrm{i.i.d.}}{\sim} \mathrm{Normal}(\theta_1+\theta_2^2, \sigma_y^2) ~~\mathrm{for}~i=1,\ldots, N.
\end{align}
Given observations $\set{y_i}_{i=1}^N$, we wish to sample the posterior distribution of $(\theta_1,\theta_2)$, a two-dimensional posterior.
In our experiments, we consider $n=100$. We employ Euclidean HMC with a step-size of $\epsilon = 0.1$; RMHMC is implemented with a step-size of 0.04; in LMC we use an integration step-size of $\epsilon= 0.1$. In each case, we set $k_\mathrm{max} = 10$. For the geometric methods, the sum of the Fisher information and the negative Hessian of the log-prior is used as a Riemannian metric. The reason that a smaller integration step-size is employed in RMHMC is that \cref{geometric-ergodicity:eq:intermediate-momentum} in the defining involution will fail to have a solution for $\breve{p}$ for large step-sizes; this necessitates the use of a smaller integration step.

In \cref{geometric-ergodicity:fig:banana-ergodicity}, we visualize the statistic $\abs{\mathrm{MMD}_u^2}$ over one-hundred steps of the Markov chain for euclidean HMC (EHMC), RMHMC and LMC; we color the RMHMC and LMC Markov chains according to how aggressively they mix with MMALA. Neither EHMC nor RMHMC exhibit clear evidence of geometric ergodicity in this target distribution; on the other hand, there exist a broad range of mixing probabilities for which LMC exhibits a linear decrease, as a function of $n$, in $\abs{\mathrm{MMD}_u^2}$ on a logarithmic scale. In \cref{geometric-ergodicity:fig:banana-metrics} we visualize the ESJD, the minimum ESS, and the minimum ESS per second for the three MCMC algorithms. We observe that LMC exhibits by far the strongest performance on these metrics, while RMHMC languishes due to its small step-size and the pathologies of the generalized leapfrog integrator applied to this posterior.

\subsection{Hierarchical Bayesian Logistic Regression}

\begin{figure}[t!]
  \begin{subfigure}[t]{0.3\textwidth}
    \centering
    \includegraphics[width=\textwidth]{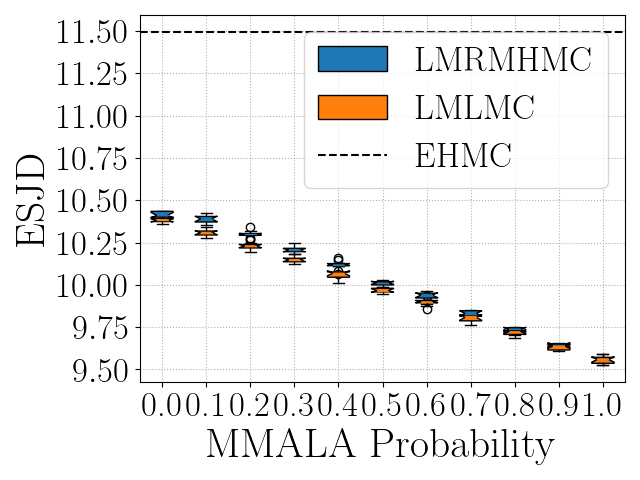}
    \caption{ESJD}
  \end{subfigure}
  ~
  \begin{subfigure}[t]{0.3\textwidth}
    \centering
    \includegraphics[width=\textwidth]{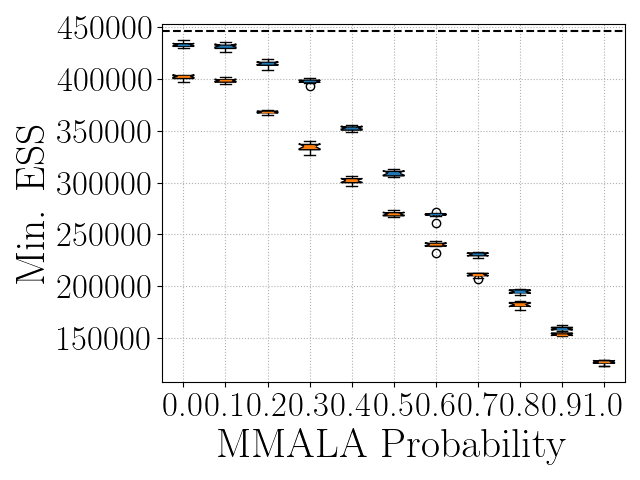}
    \caption{Min. ESS}
  \end{subfigure}
  ~
  \begin{subfigure}[t]{0.3\textwidth}
    \centering
    \includegraphics[width=\textwidth]{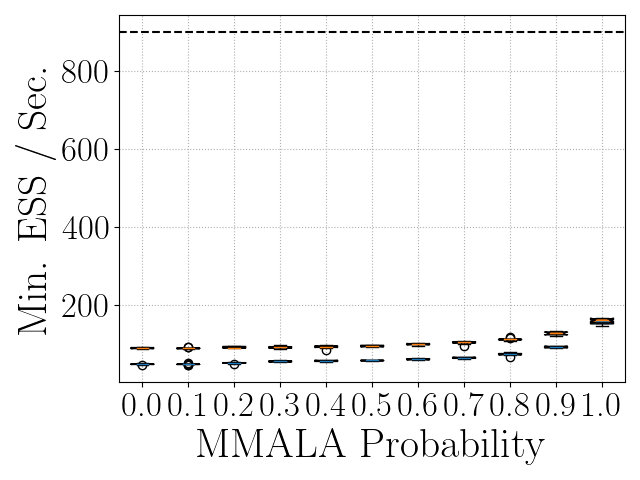}
    \caption{Min. ESS / Sec.}
  \end{subfigure}
  \caption{We show the ESJD, the minimum ESS, and the time-normalized minimum ESS for inference in the hierarchical Bayesian logistic regression posterior distribution. We observe that both of the geometric methods do not perform as well as standard HMC in this inference task. For each mixing weight with MMALA, we observe that RMHMC enjoys a greater ESJD and minimum ESS when compared against LMC; however, as LMC is a more computationally expedient procedure, LMC has a greater time-normalized ESS. In all of these metrics, however, EHMC dominates the geometric inference algorithms.}
  \label{geometric-ergodicity:fig:logistic-metrics}
\end{figure}

We consider sampling from the hierarchical Bayesian logistic regression model
\begin{align}
    \alpha &\sim \mathrm{Gamma}(\omega, \theta) \\
    \beta_i \vert \alpha &\overset{\mathrm{i.i.d.}}{\sim} \mathrm{Normal}(0, \alpha^{-1})  ~~\mathrm{for}~ i = 1,\ldots, m \\
    y_i \vert x_i,\beta &\overset{\mathrm{indep.}}{\sim} \mathrm{Bernoulli}\paren{\frac{1}{1 + \exp(-x_i^\top \beta)}} ~~\mathrm{for}~ i = 1,\ldots, N.
\end{align}
We consider a logistic regression dataset with $N = 270$ observations and $m=14$ covariates. We set $\omega=10$ and $\theta = 2$ in our experiments. We employ a Metropolis-within-Gibbs sampling procedure wherein we alternate between sampling the posterior distributions $\alpha \vert \beta, k, \theta$ and $\beta\vert \alpha, \set{(x_i,y_i)}_{i=1}^n$; sampling the former can be performed analytically, whereas we employ EHMC, RMHMC, and LMC to sample the latter. In Euclidean HMC, we set $\epsilon = 0.1$; in RMHMC and LMC we set $\epsilon = 0.8$. We set $k_\mathrm{max} = 10$ as the upper bound on the number of integration steps in each case. For implementing both LMRMHMC and LMLMC, we use the sum of the Fisher information and the negative Hessian of the log-prior as a metric.

In \cref{geometric-ergodicity:fig:logistic-metrics} we show the ESJD, the minimum ESS, and the minimum ESS per second for the hierarchical Bayesian logistic regression posterior. Neither of the geometric methods perform well in this posterior, consistently under-performing EHMC. A criticism of LMC is that its performance degrades significantly in higher dimensions \citep{geometric-foundations}; one sees evidence of this phenomenon in the smaller ESS generated by LMC; however, the computational savings due to eliminating the fixed point iterations still allow LMC to edge out a stronger time-normalized ESS compared to RMHMC.

\subsection{Neal's Funnel Distribution}

\begin{figure}[t!]
  \begin{subfigure}[t]{0.45\textwidth}
    \centering
    \includegraphics[width=\textwidth]{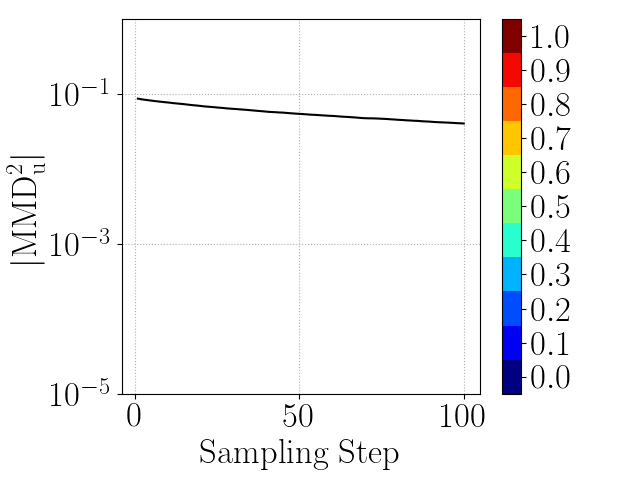}
    \caption{EHMC}
  \end{subfigure}
  ~
  \begin{subfigure}[t]{0.45\textwidth}
    \centering
    \includegraphics[width=\textwidth]{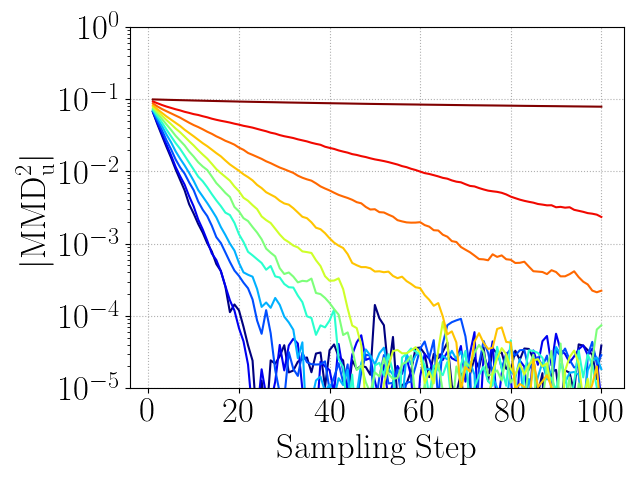}
    \caption{LMRMHMC}
  \end{subfigure}
  \caption{We examine the ergodicity of the Markov chains as a function of the sampling step in Neal's funnel distribution. We observe that EHMC struggles in this distribution due to the multiple spatial scales inherent in the funnel-shaped distribution. Moreover, we observe that aggressively mixing RMHMC with SMALA causes convergence in MMD to become increasingly slow, with convergence of the pure SMALA algorithm being slower than EHMC. On the other hand, a modest amount of mixing with SMALA produces ergodicity that is nearly indistinguishable from the unmodified implementation of RMHMC.}
  \label{geometric-ergodicity:fig:neal-funnel-ergodicity}
\end{figure}

\begin{figure}[t!]
  \begin{subfigure}[t]{0.3\textwidth}
    \centering
    \includegraphics[width=\textwidth]{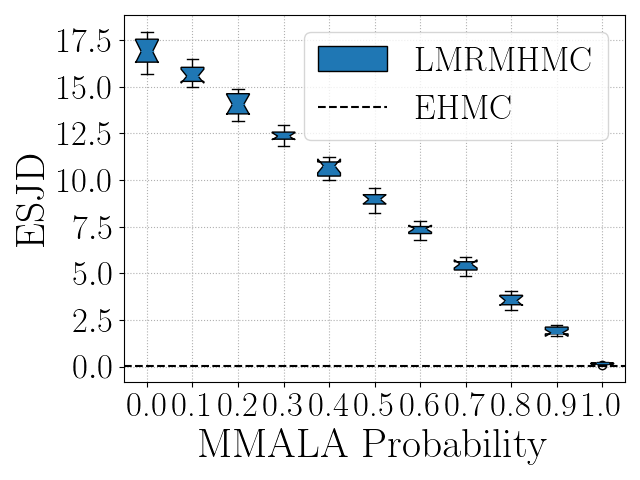}
    \caption{ESJD}
  \end{subfigure}
  ~
  \begin{subfigure}[t]{0.3\textwidth}
    \centering
    \includegraphics[width=\textwidth]{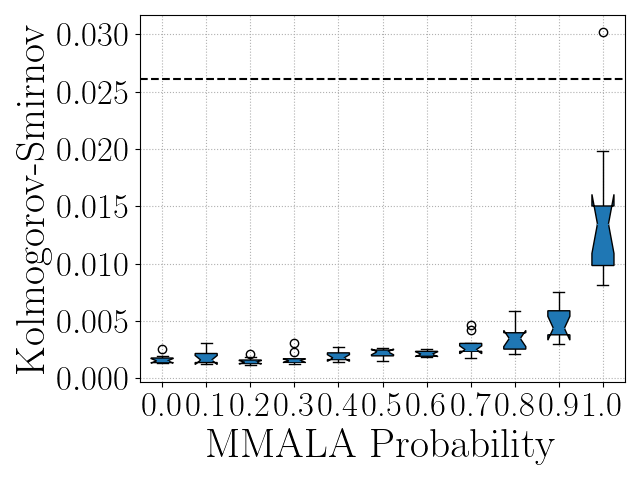}
    \caption{KS}
  \end{subfigure}
  ~
  \begin{subfigure}[t]{0.3\textwidth}
    \centering
    \includegraphics[width=\textwidth]{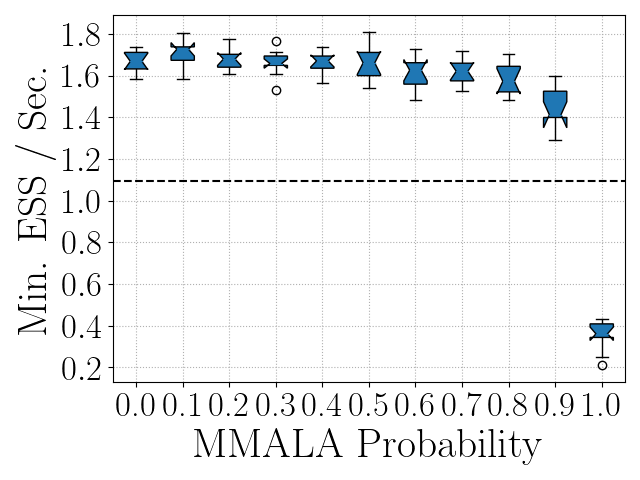}
    \caption{Min. ESS / Sec.}
  \end{subfigure}
  \caption{We show the ESJD, the distribution of Kolmogorov-Smirnov statistics, and the time-normalized minimum ESS for inference in Neal's funnel distribution. Although RMHMC and EHMC have comparable time-normalized ESS, it is clear from the distribution of KS statistics that RMHMC produces samples that are closer to the target distribution in EHMC.}
  \label{geometric-ergodicity:fig:neal-funnel-metrics}
\end{figure}

Neal's funnel distribution \citep{10.1214/aos/1056562461} is a density defined in the following hierarchical manner.
\begin{align}
    v &\sim\mathrm{Normal}(0, 9) \\
    x_i\vert v &\overset{\mathrm{i.i.d.}}{\sim} \mathrm{Normal}(0, \exp(-v)) ~~\mathrm{for}~ i=1,\ldots, N.
\end{align}
This distribution is shaped like a funnel, in which the thickness of the ``neck'' is being controlled by the random variable $v$. This model is reflective of posteriors encountered in hierarchical models with sparse data. The objective in this task is to jointly sample $(v, x_1,\ldots, x_N)$, producing a $(N+1)$-dimensional target distribution. In Euclidean HMC we employ an integration step-size of $\epsilon=0.1$ and $k_\mathrm{max} = 10$. Our implementation of RMHMC uses $k_\mathrm{max} = 20$ integration steps with a step-size of $\epsilon=0.1$ with the SoftAbs metric. We do not consider LMC in this task since we found it non-obvious how the SoftAbs structure could be extended into the LMC framework while preserving the cubic computational cost at each step. In this experiment, we use the SoftAbs metric.

In \cref{geometric-ergodicity:fig:neal-funnel-ergodicity}, we visualize $\abs{\mathrm{MMD}_u^2}$ for both EHMC and RMHMC. As expected, EHMC struggles to sample from Neal's funnel distribution due to the multiscale phenomena. On the other hand, RMHMC exhibits much stronger convergence properties, having a linear decrease over several possible mixing probabilities with SMALA. We find that aggressively mixing with SMALA can be counter-productive, however, due to the less efficient traversal of the target distribution by single-step methods. In \cref{geometric-ergodicity:fig:neal-funnel-metrics} we show the ESJD, the Kolmogorov-Smirnov statistics, and the minimum ESS per second. On all of these metrics, RMHMC clearly outperforms EHMC.

\subsection{Stochastic Volatility Model}

\begin{figure}[t!]
  \begin{subfigure}[t]{0.3\textwidth}
    \centering
    \includegraphics[width=\textwidth]{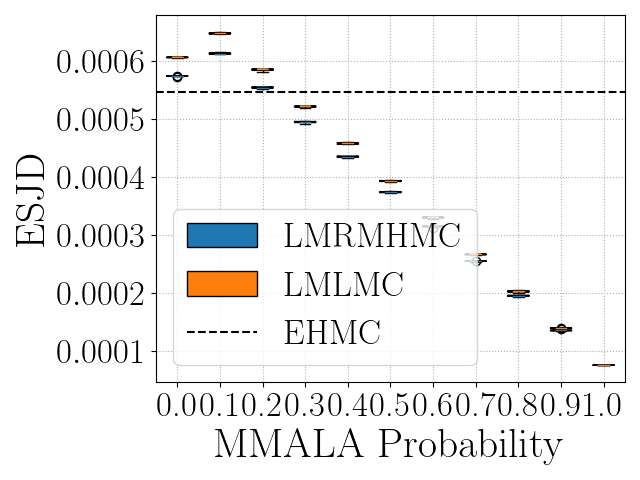}
    \caption{ESJD}
  \end{subfigure}
  ~
  \begin{subfigure}[t]{0.3\textwidth}
    \centering
    \includegraphics[width=\textwidth]{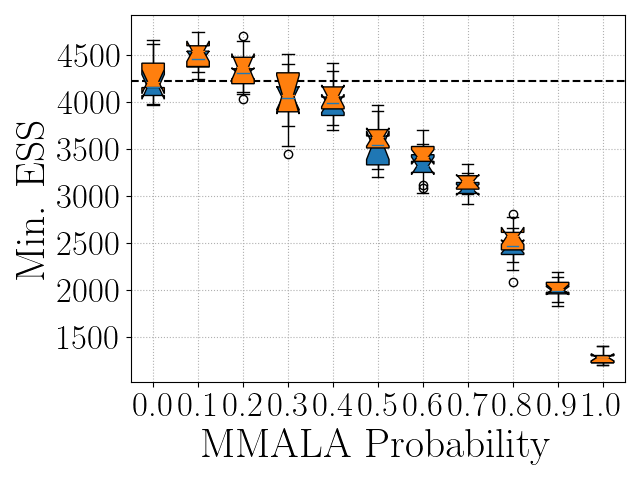}
    \caption{Min. ESS}
  \end{subfigure}
  ~
  \begin{subfigure}[t]{0.3\textwidth}
    \centering
    \includegraphics[width=\textwidth]{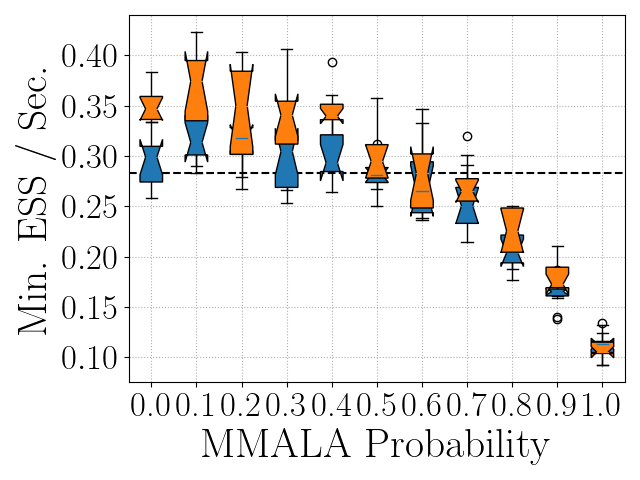}
    \caption{Min. ESS / Sec.}
  \end{subfigure}
  \caption{We show the ESJD, the minimum ESS, and the time-normalized minimum ESS for inference in the posterior distribution of the stochastic volatility model. We observe that both of the geometric methods with a modest mixing probability with MMALA produce time-normalized ESS that are competitive with EHMC. For each mixing weight with MMALA, we observe that LMC enjoys a greater ESJD when compared against RMHMC.}
  \label{geometric-ergodicity:fig:stochastic-volatility-metrics}
\end{figure}

We consider a stochastic volatility model with the following generative model.
\begin{align}
    \frac{\phi+1}{2} & \sim \mathrm{Beta}(20, 3/2) \\
    \sigma^2 &\sim \mathrm{InverseChiSquared}(10, 1/20) \\
    x_1 \vert \phi, \sigma^2 &\sim \mathrm{Normal}(0, \sigma^2 / (1-\phi^2)) \\
    x_{t+1} \vert x_{t},\phi,\sigma^2 &\sim \mathrm{Normal}(\phi x_t, \sigma^2) ~~\mathrm{for}~t=2,\ldots, T-1 \\
    y_t \vert \beta, x_t &\sim \mathrm{Normal}(0, \beta^2 \exp(x_t)) ~~\mathrm{for}~t=1,\ldots, T.
\end{align}
Additionally, the parameter $\beta$ is equipped with an improper prior proportional to $\beta^{-1}$. In this example, we seek to generate samples from the joint distribution $x_1,\ldots, x_T,\beta,\phi,\sigma^2$ given observations $y_1,\ldots, y_T$. We employ Metropolis-within-Gibbs-like strategy wherein we alternate between sampling $(x_1, \ldots, x_T)\vert (y_1,\ldots, y_T) \phi,\sigma^2, \beta$ and $(\phi,\sigma^2,\beta)\vert \set{(x_i,y_i)}_{i=1}^T$; in the former case we employ Euclidean HMC whereas in the latter case we compare Euclidean HMC, RMHMC, and LMC. In sampling either distribution, the Riemannian metric is chosen as the sum of the Fisher information and the negative Hessian of the log-prior. In our experiments we set $T=1,000$. When using Euclidean HMC to sample $(\phi,\sigma^2,\beta)$ we use $k_\mathrm{max} = 50$ integration steps and a step-size of $\epsilon = 0.01$; in the case of RMHMC and LMC we use $k_\mathrm{max} = 6$ integration steps and a step-size of $\epsilon  = 0.5$. 

In \cref{geometric-ergodicity:fig:stochastic-volatility-metrics} we show the ESJD, the minimum ESS, and the minimum ESS per second. We observe that employing a modest mixture probability with MMALA produces a Markov chain that is marginally better than EHMC.

\subsection{Log-Gaussian Cox-Poisson Process}

\begin{figure}[t!]
  \begin{subfigure}[t]{0.3\textwidth}
    \centering
    \includegraphics[width=\textwidth]{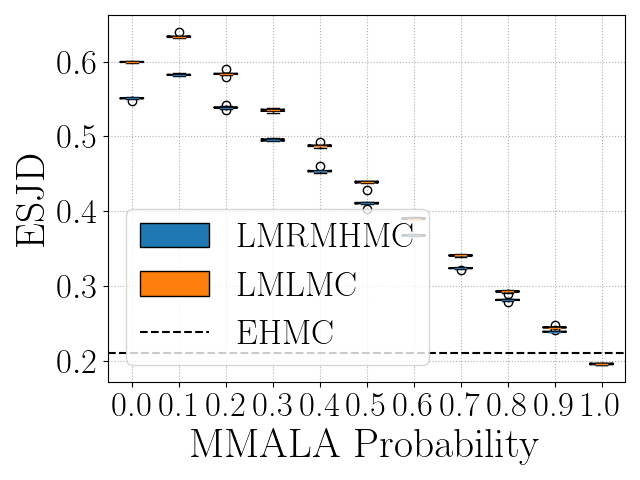}
    \caption{ESJD}
  \end{subfigure}
  ~
  \begin{subfigure}[t]{0.3\textwidth}
    \centering
    \includegraphics[width=\textwidth]{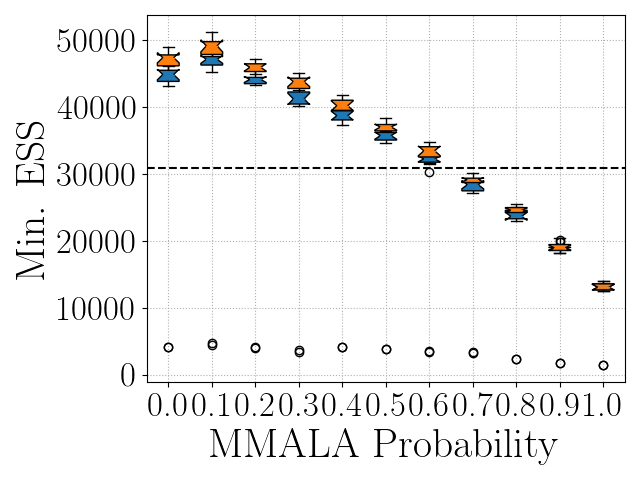}
    \caption{Min. ESS}
  \end{subfigure}
  ~
  \begin{subfigure}[t]{0.3\textwidth}
    \centering
    \includegraphics[width=\textwidth]{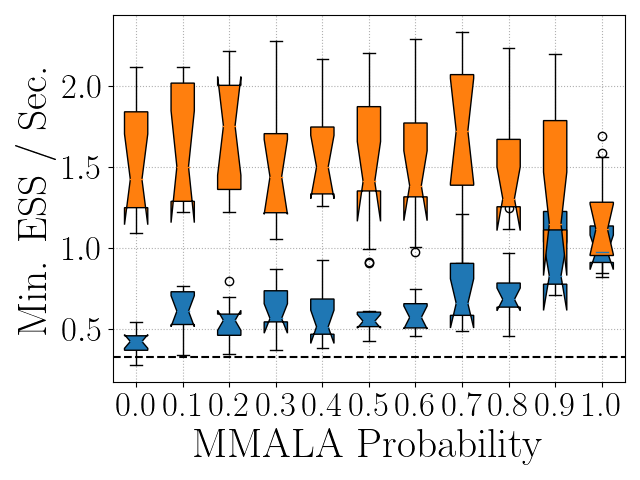}
    \caption{Min. ESS / Sec.}
  \end{subfigure}
  \caption{We show the ESJD, the minimum ESS, and the time-normalized minimum ESS for inference in the posterior distribution of the log-Gaussian Cox-Poisson model. For any mixing probability, the geometric methods enjoy larger time-normalized effective sample sizes than EHMC, with LMC outperforming RMHMC on these metrics. For nearly every mixing probability, we additionally observe that the ESJD is greater for LMC and RMHMC than for EHMC.}
  \label{geometric-ergodicity:fig:cox-poisson-metrics}
\end{figure}

We consider inference in a log-Gaussian Cox-Poisson model with the following generative model:
\begin{align}
  \Sigma_{(i, j), (i', j')} \vert \sigma^2, \beta &= \sigma^2 \exp\paren{-\sqrt{(i-i')^2 + (j-j')^2} / (N \beta)} \\
  \mathrm{vec}(\mathbf{x}) \vert \Sigma &\sim \mathrm{MultivariateNormal}(\mu\mathbf{1}, \Sigma) \\
  y_{ij}\vert x_{ij} &\sim \mathrm{Poisson}(\exp(x_{ij}) / N^2),
\end{align}
with priors $\beta \sim\mathrm{Gamma}(2, 1/2)$ and $\sigma^2\sim \mathrm{Gamma}(2, 1/2)$. In this example, the objective is to sample the joint distribution $(\beta,\sigma^2,\set{x_{ij}}_{i,j=1}^N)$ given observations $\set{y_{ij}}_{i,j=1}^N$. As in the case of stochastic volatility model, we alternatively sample between $\set{x_{ij}}_{i,j=1}^N$ given $\set{y_{ij}}_{i,j=1}^N$, $\sigma^2$, and $\beta$, and $(\beta,\sigma^2)$ given $\set{x_{ij}}_{i,j=1}^N$. In each case, the metric is given by the sum of the Fisher information and the negative Hessian of the log-prior. In our experiments we set $N=16$. In the former case we employ Euclidean HMC, whereas in the latter case we compare Euclidean HMC, RMHMC and LMC. When implementing Euclidean HMC we employ $k_\mathrm{max} = 50$ integration steps and a step-size of $\epsilon = 0.01$; in RMHMC and LMC we use $k_\mathrm{max} = 6$ integration steps and a step-size of $\epsilon = 0.5$.

In \cref{geometric-ergodicity:fig:cox-poisson-metrics} we show the ESJD, the minimum ESS, and the minimum ESS per second. We observe that there is a range of mixture probabilities for which mixing with the MMALA in LMRMHMC and LMLMC exhibit superior performance compared to EHMC. Both RMHMC and LMC exhibit similar minimum ESS metrics, but due to its computational advantage, LMC produces a larger minimum ESS per second.

\subsection{Fitzhugh-Nagumo Model}

\begin{figure}[t!]
  \begin{subfigure}[t]{0.3\textwidth}
    \centering
    \includegraphics[width=\textwidth]{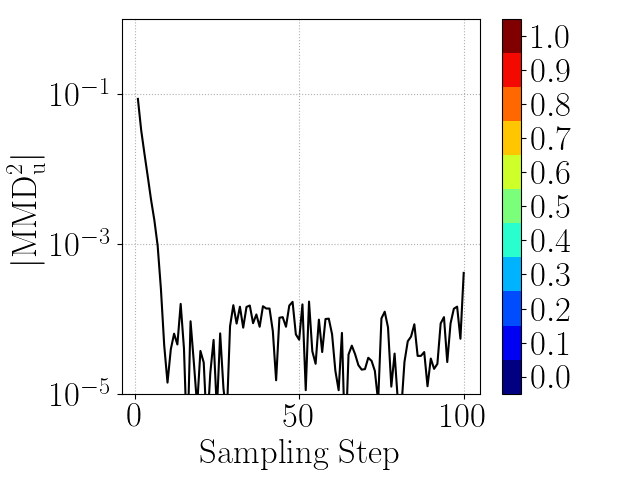}
    \caption{EHMC}
  \end{subfigure}
  ~
  \begin{subfigure}[t]{0.3\textwidth}
    \centering
    \includegraphics[width=\textwidth]{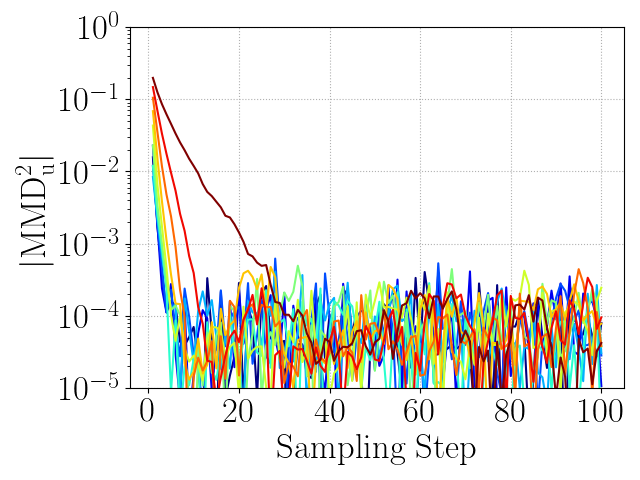}
    \caption{LMRMHMC}
  \end{subfigure}
  ~
  \begin{subfigure}[t]{0.3\textwidth}
    \centering
    \includegraphics[width=\textwidth]{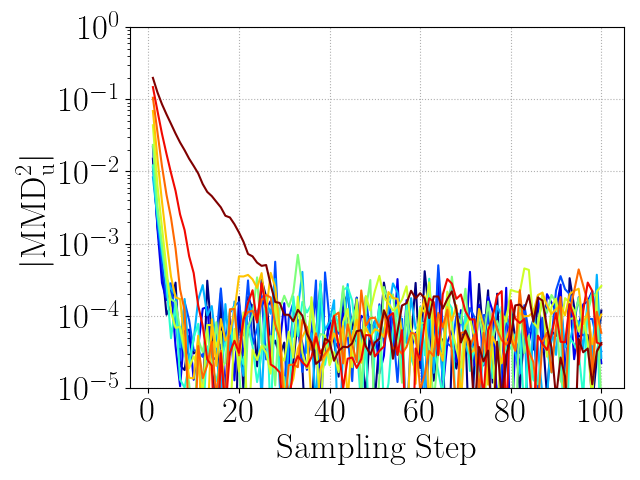}
    \caption{LMLMC}
  \end{subfigure}
  \caption{We examine the ergodicity of the Markov chains as a function of the sampling step in the Fitzhugh-Nagumo posterior distribution. We observe that when a modest mixing probability is employed, the modified RMHMC and LMC transition kernels mix more efficiently than the EHMC transition kernel, being nearly indistinguishable from the unmodified RMHMC and LMC transition kernels. When one uses a large mixing probability, we see that the mixing rate is decreased.}
  \label{geometric-ergodicity:fig:fitzhugh-nagumo-ergodicity}
\end{figure}

\begin{figure}[t!]
  \begin{subfigure}[t]{0.3\textwidth}
    \centering
    \includegraphics[width=\textwidth]{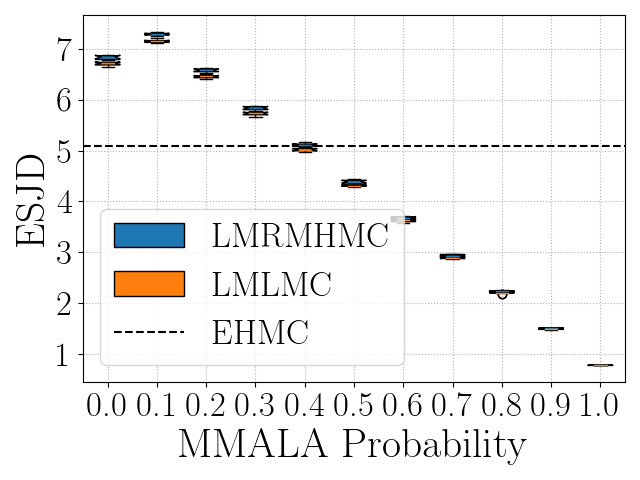}
    \caption{ESJD}
  \end{subfigure}
  ~
  \begin{subfigure}[t]{0.3\textwidth}
    \centering
    \includegraphics[width=\textwidth]{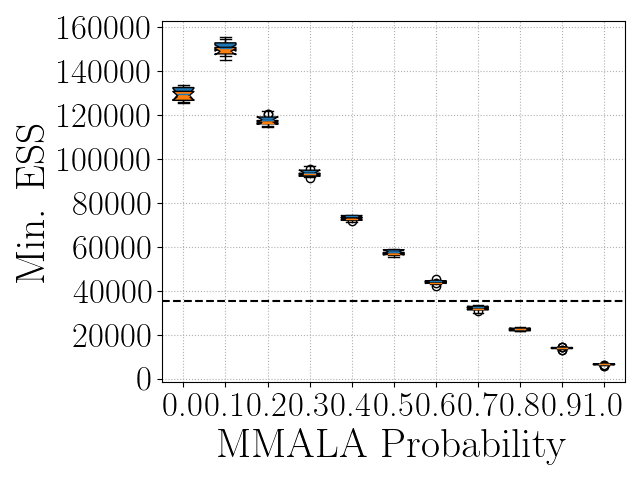}
    \caption{Min. ESS}
  \end{subfigure}
  ~
  \begin{subfigure}[t]{0.3\textwidth}
    \centering
    \includegraphics[width=\textwidth]{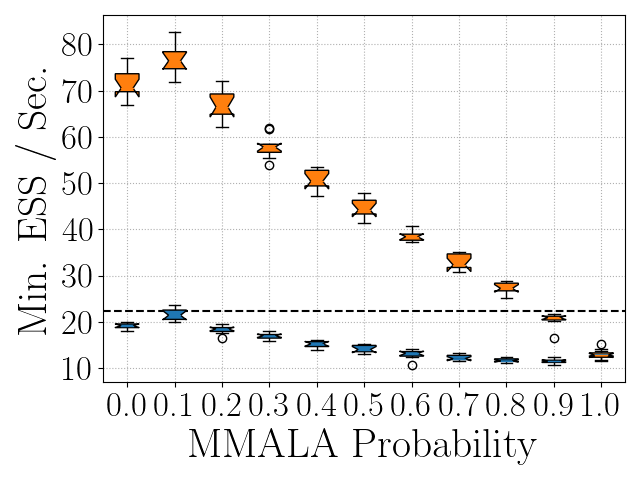}
    \caption{Min. ESS / Sec.}
  \end{subfigure}
  \caption{We show the ESJD, the minimum ESS, and the time-normalized minimum ESS for inference in the Fitzhugh-Nagumo posterior distribution. LMC dominates both RMHMC and EHMC in this inference task under the time-normalized ESS metrics with its greater computational expediency. When timing is not accounted for, we observe that RMHMC and LMC enjoy similar ESJD and minimum ESS metrics. We observe that for both RMHMC and LMC, aggressively mixing with MMALA causes the ESJD distance to decrease; this correspondingly produces a decrease in the minimum ESS.}
  \label{geometric-ergodicity:fig:fitzhugh-nagumo-metrics}
\end{figure}

Given $\R$-valued parameters $a$, $b$, and $c$, the Fitzhugh-Nagumo differential equations are defined by,
\begin{align}
  \dot{v}_t &= c\paren{v_t - \frac{v_t^3}{3} + r_t} \\
  \dot{r}_t &= -\paren{\frac{v_t - a + b r_t}{c}}.
\end{align}
Given initial conditions $v_0$ and $r_0$, we consider the following generative model:
\begin{align}
    (a, b, c) &\overset{\mathrm{i.i.d.}}{\sim} \mathrm{Normal}(0, 1) \\
    \hat{r}_{t_k}\vert a, b, c, r_{t_k}, \sigma^2 &\overset{\mathrm{indep.}}{\sim} \mathrm{Normal}(r_{t_k}, \sigma^2)~~\mathrm{for}~ k = 1,\ldots, n \\
    \hat{v}_{t_k} \vert a, b, c, v_{t_k}, \sigma^2 &\overset{\mathrm{indep.}}{\sim} \mathrm{Normal}(v_{t_k}, \sigma^2) ~~\mathrm{for}~ k=1,\ldots, n,
\end{align}
where $t_1,\ldots, t_n$ are equally spaced points between $[0, T]$. In our experiments, we set $v_0 = 1$, $r_0 = -1$, $\sigma^2=1/4$, $T=10$, and $n = 200$. Our metric is given by the sum of the Fisher information and negative Hessian of the log-prior. In our implementations, we use Euclidean HMC with a step-size of $\epsilon=0.01$ and $k_\mathrm{max} = 10$ integration steps. In RMHMC and LMC, we employ an integration step-size of $\epsilon=0.5$ and $k_\mathrm{max} = 6$ integration steps.

\Cref{geometric-ergodicity:fig:fitzhugh-nagumo-ergodicity} shows the ergodicity measures for EHMC, RMHMC, and LMC. We observe that each of these MCMC algorithms exhibit a linear decrease in $\abs{\mathrm{MMD}_u^2}$ on a logarithmic scale. In \cref{geometric-ergodicity:fig:fitzhugh-nagumo-metrics} we visualize the ESJD, the minimum ESS, and the minimum ESS per second for the three MCMC algorithms. We observe that LMC exhibits the strongest performance in terms of the minimum ESS per second, whereas RMHMC only approaches the time-normalized performance of EHMC due to its complexity.

\subsection{Multi-Scale Student Distribution}\label{subsec:student-t}

\begin{figure}[t!]
  \begin{subfigure}[t]{0.3\textwidth}
    \centering
    \includegraphics[width=\textwidth]{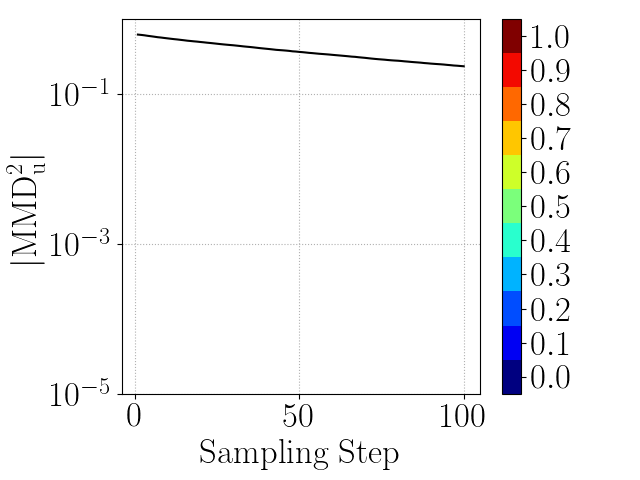}
    \caption{EHMC}
  \end{subfigure}
  ~
  \begin{subfigure}[t]{0.3\textwidth}
    \centering
    \includegraphics[width=\textwidth]{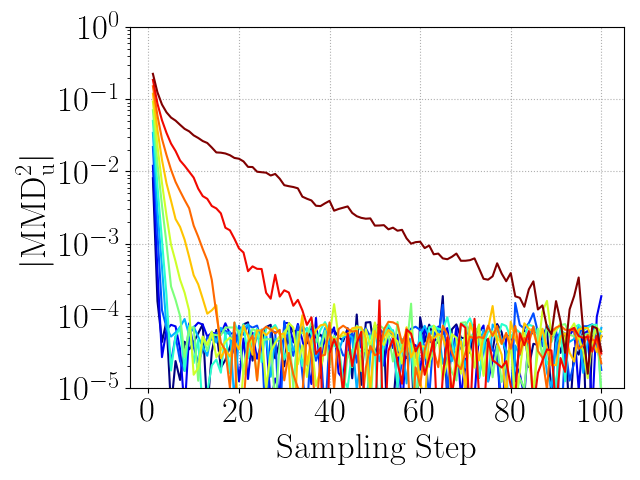}
    \caption{LMRMHMC}
  \end{subfigure}
  ~
  \begin{subfigure}[t]{0.3\textwidth}
    \centering
    \includegraphics[width=\textwidth]{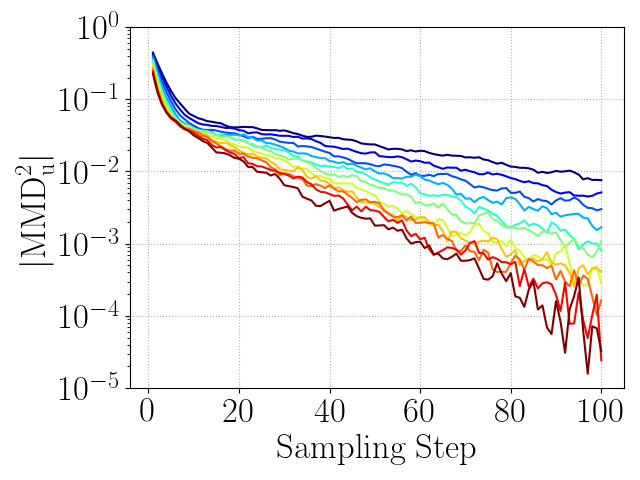}
    \caption{LMLMC}
  \end{subfigure}
  \caption{We examine the ergodicity of the Markov chains as a function of the sampling step in the multiscale Student-$t$ distribution. Here we see that, in the case of LMC, it is beneficial from an ergodicity perspective to mix heavily with MMALA, whereas for RMHMC, only modest mixing probabilities can produce a modified Markov chain that is competitive with the unmodified version. The Euclidean HMC struggles in this posterior distribution due to the multiple spatial scales.}
  \label{geometric-ergodicity:fig:t-ergodicity}
\end{figure}

\begin{figure}[t!]
  \begin{subfigure}[t]{0.45\textwidth}
    \centering
    \includegraphics[width=\textwidth]{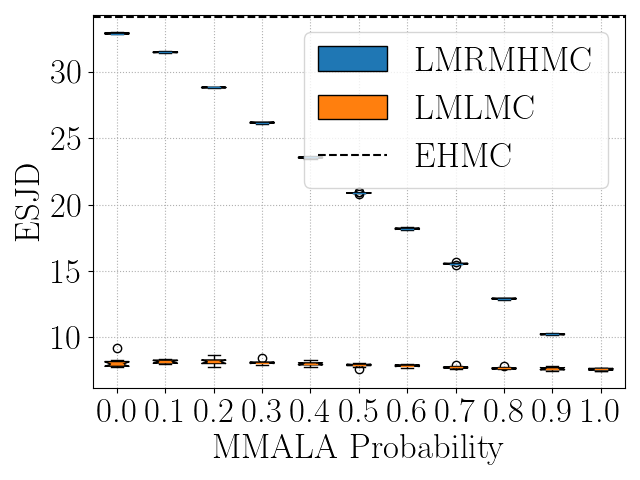}
    \caption{ESJD}
  \end{subfigure}
  ~
  \begin{subfigure}[t]{0.45\textwidth}
    \centering
    \includegraphics[width=\textwidth]{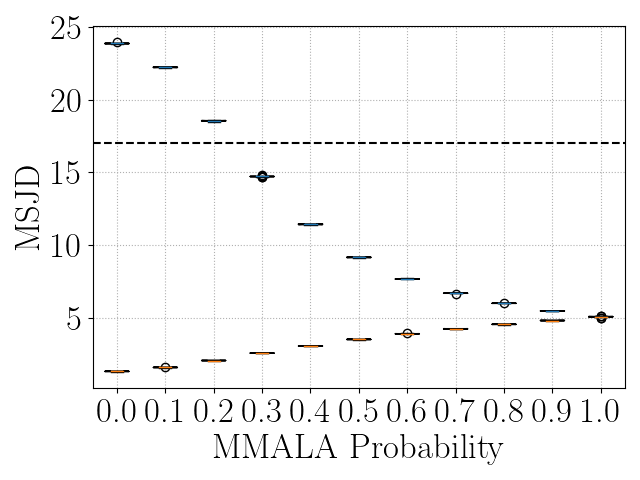}
    \caption{MSJD}
  \end{subfigure}
  \caption{We show the ESJD and the MSJD for inference in the multiscale Student-$t$ distribution. Curiously, EHMC enjoys the largest ESJD, but this does not correspond to more efficient sampling, as shown in the inferior distribution of KS statistics. To examine this property further, we also measured the MSJD, which reveals that EHMC is less effective in traversing the parameter space compared to the RMHMC geometric method.}
  \label{geometric-ergodicity:fig:t-jump}
\end{figure}

\begin{figure}[t!]
  \begin{subfigure}[t]{0.45\textwidth}
    \centering
    \includegraphics[width=\textwidth]{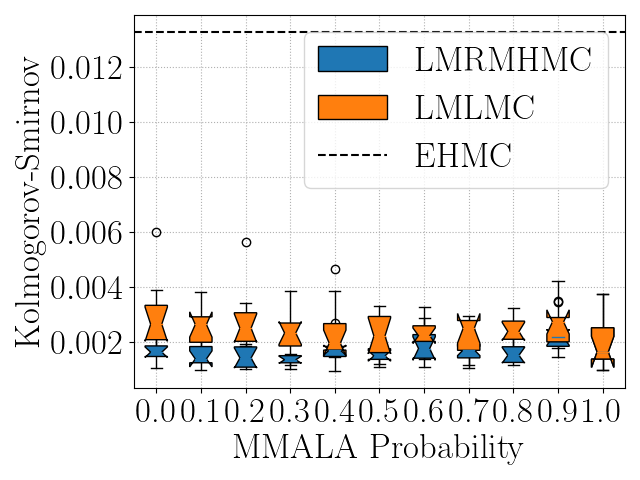}
    \caption{KS}
  \end{subfigure}
  ~
  \begin{subfigure}[t]{0.45\textwidth}
    \centering
    \includegraphics[width=\textwidth]{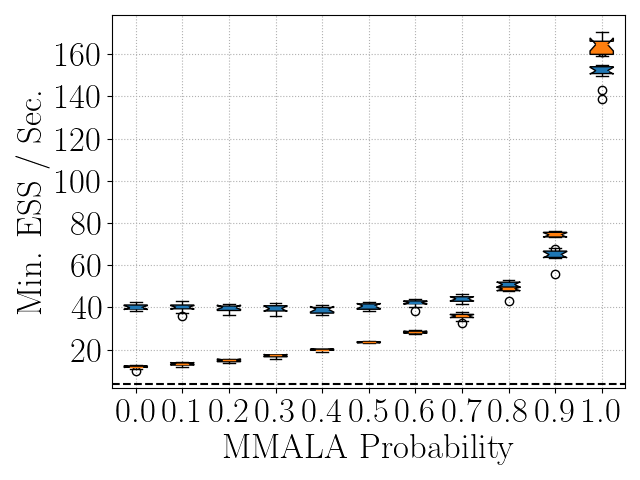}
    \caption{Min. ESS / Sec.}
  \end{subfigure}
  \caption{We show the distribution of Kolmogorov-Smirnov statistics and the time-normalized minimum ESS for inference in the multiscale Student-$t$ distribution. For the geometric methods, we observe that RMHMC has much larger ESJD than LMC; however, in terms of time-normalized performance, it is best to fully mix either method with MMALA. Indeed, we find that MMALA on its own produces a distribution of KS statistics that is competitive with any of the mixture kernels.}
  \label{geometric-ergodicity:fig:t-metrics}
\end{figure}

Fix $m\in\mathbb{N}$ and let $\Sigma\in \mathrm{PD}(m)$ and $\nu > 2$. The density function of the multivariate Student-$t$ distribution is,
\begin{align}
    \pi(q) \propto \paren{1 + \frac{1}{\nu} x^\top \Sigma^{-1} x}^{-(m + \nu) / 2}.
\end{align}
In our experiments we consider $m=20$, $\nu = 5$, and $\Sigma = \mathrm{diag}(1,\ldots, 1, 10^4)\in\R^{m\times m}$. The presence of severely differing spatial scales in the multivariate Student-$t$ distribution will cause Euclidean HMC with identity mass matrix to exhibit highly oscillatory behavior \citep{pourzanjani2019implicit}, which will limit the efficiency of the method both in terms of ergodicity and the effective sample size. We employ Euclidean HMC with a step-size of $\epsilon=0.8$, and RMHMC and LMC with a step-size of $\epsilon=0.7$; in each case we employ $k_\mathrm{max} = 20$ integration steps. In the case of RMHMC and LMC, we consider a metric given by the positive definite term in the Hessian of the log-density of the multivariate Student-$t$ distribution. 

In \cref{geometric-ergodicity:fig:t-ergodicity} we show $\abs{\mathrm{MMD}_u^2}$ as a function of the number of Markov chain steps. All of the MCMC algorithms produce a linear decrease in the estimate of the maximum mean discrepancy on a logarithmic scale; however, the methods differ drastically in terms of the {\it slope} of this linear relationship. We see that RMHMC exhibits by far the fastest convergence. Notably, LMC exhibits {\it slower} convergence on this target distribution than MMALA; this is due to the dimensionality of the posterior, in which LMC struggles to maintain the Hamiltonian energy required to accept proposals. In \cref{geometric-ergodicity:fig:t-metrics} we show the ESJD, KS, and the minimum ESS per second. EHMC produces the largest ESJD but this does not translate into a large ESS due to the oscillatory behavior of the EHMC proposal mechanism, with RMHMC producing the largest minimum ESS per second despite its computational complexity.

\section{Conclusion}

This work has considered methods by which to equip RMHMC and LMC with a geometric ergodicity theory. The fundamental technique we adopt is to replace the RMHMC (or LMC) transition kernel consisting of a single integration step with the MMALA transition kernel. This modification is inspired by the Euclidean case, in which single-step HMC and MALA can be constructed to be exactly equivalent. By establishing reversibility of the marginal transition kernels, geometric ergodicity can be inherited from MMALA. We evaluated the modified variations of RMHMC and LMC, called LMRMHMC and LMLMC, respectively, on a suite of Bayesian inference tasks. We found that aggressively mixing with MMALA transition kernel can be detrimental for the performance of the Markov chain on a variety of metrics, but that more modest mixing can produce behaviors competitive with, or exceeding, the original RMHMC or LMC methods while still imbuing the methods with a supporting theory of geometric ergodicity.

\backmatter

\bmhead{Acknowledgments}

We thank the Yale Center for Research Computing for use of the research computing infrastructure.
This material is based upon work supported by the National Science Foundation Graduate Research Fellowship under Grant No. 1752134. Any opinion, findings, and conclusions or recommendations expressed in this material are those of the authors(s) and do not necessarily reflect the views of the National Science Foundation. The work is also supported in part by NIH/NIGMS R01GM136780 and AFOSR FA9550-21-1-0317.

\section*{Statements and Declarations}

No competing interests to declare.

\begin{appendices}

\section{Proofs}

\subsection{Proof of \Cref{geometric-ergodicity:prop:mala-hmc-equivalence}}\label{geometric-ergodicity:app:proof-mala-hmc-equivalence}

\begin{proof}
The proposal distribution of MALA is
\begin{align}
    \tilde{q}\vert q &\sim \mathrm{Normal}\paren{q + \frac{\epsilon^2}{2} \mathbf{A}^{-1}\nabla \log \pi(q), \epsilon^2 \mathbf{A}^{-1}} \\
    \tilde{\pi}(\tilde{q}\vert q) &\propto \exp\paren{-\frac{1}{2\epsilon^2} \paren{\tilde{q} - q - \frac{\epsilon^2}{2}\mathbf{A}^{-1} \nabla \log\pi(q)}^\top \mathbf{A} \paren{\tilde{q} - q - \frac{\epsilon^2}{2} \mathbf{A}^{-1} \nabla \log\pi(q)}} \\
    &= \exp\paren{-\frac{1}{2\epsilon^2} \paren{\mathbf{A}\tilde{q} - \mathbf{A}q - \frac{\epsilon^2}{2} \nabla \log\pi(q)}^\top \mathbf{A}^{-1} \paren{\mathbf{A}\tilde{q} - \mathbf{A}q - \frac{\epsilon^2}{2} \nabla \log\pi(q)}}.
\end{align}
Therefore, the acceptance probability of MALA is,
\begin{align}
\begin{split}
    \label{geometric-ergodicity:eq:mala-acceptance} &\min\set{1, \frac{\pi(\tilde{q}) \tilde{\pi}(q\vert \tilde{q})}{\pi(q) \tilde{\pi}(\tilde{q}\vert q)}} \\
    =& \min\set{1, \frac{\pi(\tilde{q})}{\pi(q)}\cdot \frac{\exp\paren{-\frac{1}{2\epsilon^2} \paren{\mathbf{A}q - \mathbf{A}\tilde{q} - \frac{\epsilon^2}{2} \nabla \log\pi(\tilde{q})}^\top \mathbf{A}^{-1} \paren{\mathbf{A}q - \mathbf{A}\tilde{q} - \frac{\epsilon^2}{2} \nabla \log\pi(\tilde{q})}}}{\exp\paren{-\frac{1}{2\epsilon^2} \paren{\mathbf{A}\tilde{q} - \mathbf{A}q - \frac{\epsilon^2}{2} \nabla \log\pi(q)}^\top \mathbf{A}^{-1} \paren{\mathbf{A}\tilde{q} - \mathbf{A}q - \frac{\epsilon^2}{2} \nabla \log\pi(q)}}}}.
\end{split}
\end{align}
The proposal $\tilde{q}$ given $q$ can be sampled by generating $z\sim\mathrm{Normal}(0, \mathrm{Id})$ and setting
\begin{align}
    \label{geometric-ergodicity:eq:mala-proposal} \tilde{q}=q+\frac{\epsilon}{2} \mathbf{A}^{-1}\nabla \log\pi(q) + \epsilon\sqrt{\mathbf{A}^{-1}} z.
\end{align}
In HMC, the leapfrog integrator is applied to the Hamiltonian $H(q, p) = -\log\pi(q) + \frac{1}{2} p^\top\mathbf{A}^{-1}p$. The sequence of updates is,
\begin{align}
    \bar{p} &= p + \frac{\epsilon}{2} \nabla\log\pi(q) \\
    \label{geometric-ergodicity:eq:hmc-proposal} \tilde{q}' &= q + \epsilon \mathbf{A}^{-1} \bar{p} \\
    &= q + \frac{\epsilon^2}{2}\mathbf{A}^{-1} \nabla \log \pi(q) + \epsilon \mathbf{A}^{-1} p \\
    \tilde{p}' &= \bar{p} + \frac{\epsilon}{2} \nabla\log\pi(\tilde{q}') \\
    &= p + \frac{\epsilon}{2} \nabla\log\pi(q) + \frac{\epsilon}{2} \nabla\log\pi(\tilde{q}').
\end{align}
When sampling from the distribution with density proportional to $\exp(-H(q, p))$, one sees by inspection that $p$ is independent of $q$ and that $p\sim\mathrm{Normal}(0, \mathbf{A})$. Therefore, using the same $z$ as in \cref{geometric-ergodicity:eq:mala-proposal}, we can sample $p$ by setting $p=\sqrt{\mathbf{A}} z$ so that $\epsilon\mathbf{A}^{-1} p = \epsilon\sqrt{\mathbf{A}^{-1}}z$; we therefore see that, in this case, HMC and MALA produce exactly the same proposal (i.e. $\tilde{q}'=\tilde{q}$). Recall that the HMC acceptance probability is,
\begin{align}
    \label{geometric-ergodicity:eq:hmc-acceptance} \min\set{1, \frac{\pi(\tilde{q})}{\pi(q)} \cdot \frac{\exp\paren{-\frac{1}{2} (\tilde{p}')^\top\mathbf{A}^{-1} (\tilde{p}')}}{\exp\paren{-\frac{1}{2} p^\top\mathbf{A}^{-1} p}}}.
\end{align}
By rearranging \cref{geometric-ergodicity:eq:hmc-proposal} we find,
\begin{align}
    p &= \frac{1}{\epsilon} \paren{\mathbf{A}\tilde{q} - \mathbf{A}q - \frac{\epsilon^2}{2} \nabla \log\pi(q)} \\
    \tilde{p}' &= \frac{1}{\epsilon} \paren{\mathbf{A}\tilde{q} - \mathbf{A}q + \frac{\epsilon^2}{2} \nabla\log\pi(\tilde{q})} \\
    &= -\frac{1}{\epsilon} \paren{\mathbf{A}q - \mathbf{A}\tilde{q} - \frac{\epsilon^2}{2} \nabla\log\pi(\tilde{q})}.
\end{align}
Substituting these into \cref{geometric-ergodicity:eq:hmc-acceptance} and comparing to \cref{geometric-ergodicity:eq:mala-acceptance} shows that not only are $\tilde{q}$ and $\tilde{q}'$ identical but that the acceptance probabilities are also identical. Therefore, the marginal chain of single-step HMC is exactly equivalent to the MALA chain.
\end{proof}

\subsection{Proof of \Cref{geometric-ergodicity:lem:marginal-transition-kernel}}\label{geometric-ergodicity:app:proof-marginal-transition-kernel}

\begin{proof}
  \begin{align}
      \mathrm{Pr}(q^{n+1}\in Q\vert q^n=q) &= \mathrm{Pr}(q^{n+1}\in Q ~\mathrm{and}~ p^{n+1}\in \R^m\vert q^n=q) \\
      &= \int_{\R^m} \mathrm{Pr}(q^{n+1}\in Q ~\mathrm{and}~ p^{n+1}\in \R^m\vert q^n=q, p^n=p) ~\pi(p\vert q)~\mathrm{d}p \\
      &= \int_{\R^m} K((q, p), (Q,\R^m)) ~\pi(p\vert q)~\mathrm{d}p.
  \end{align}
\end{proof}

\subsection{Proof of \Cref{geometric-ergodicity:prop:marginal-reversible}}\label{geometric-ergodicity:app:proof-marginal-reversible}

\begin{proof}
Given $Q, Q'\in\mathfrak{B}(\R^m)$ we have
\begin{align}
    \int_{Q} \tilde{K}(q, Q') ~\pi(q)~\mathrm{d}q &= \int_{Q} \int_{\R^m} K((q, p), (Q', \R^m)) ~\pi(p\vert q) ~\mathrm{d}p ~\pi(q)~\mathrm{d}q \\
    &= \int_{Q} \int_{\R^m} K((q, p), (Q', \R^m)) ~\pi(q, p) ~\mathrm{d}p ~\mathrm{d}q \\
    \label{geometric-ergodicity:eq:detailed-balance} &= \int_{Q'} \int_{\R^m} K((q, p), (Q, \R^m)) ~\pi(q, p) ~\mathrm{d}p ~\mathrm{d}q \\
    &= \int_{Q'} \int_{\R^m} K((q, p), (Q, \R^m)) ~\pi(p\vert q) ~\mathrm{d}p ~\pi(q)~\mathrm{d}q \\
    &= \int_{Q'} \tilde{K}(q, Q) ~\pi(q)~\mathrm{d}q,
\end{align}
where in \cref{geometric-ergodicity:eq:detailed-balance} we have used the fact that the phase space chain satisfies detailed balance with respect to $\pi(q, p)$. This verifies that the marginal chain satisfies detailed balance. 
\end{proof}

\subsection{Proof of Involution Composition}

\begin{proposition}\label{prop:involution-composition}
Suppose that $\Phi : \R^m\times\R^m\to\R^m\times\R^m$ is an invertible function and let $\mathbf{F}$ be the momentum flip function given in \cref{def:rmhmc-transition-kernel}. Suppose that $\mathbf{F}\circ\Phi$ is an involution, then $\mathbf{F}\circ\Phi^k$ is also an involution.
\end{proposition}
\begin{proof}
This will be proved by induction with the base case established by assumption. As the inductive hypothesis, assume that $\mathbf{F}\circ\Phi^k$ is an involution. Using the fact that $\mathbf{F}\circ \Phi$ is an involution, we immediately obtain that $\Phi^{-1}\circ\mathbf{F} = \mathbf{F}\circ\Phi$. Using the inductive hypothesis, one also has $\Phi^{-k}\circ\mathbf{F} = \mathbf{F}\circ\Phi^k$. Therefore, we obtain,
\begin{align}
    & \Phi^{-1}\circ\mathbf{F} = \mathbf{F}\circ\Phi \\
    \implies& \Phi^{-1}\circ\mathbf{F}\circ\Phi^k = \mathbf{F}\circ\Phi^{k+1} \\
    \implies& \Phi^{-1}\circ\Phi^{-k}\circ\mathbf{F} = \mathbf{F}\circ\Phi^{k+1} \\
    \implies& \Phi^{-k-1}\circ\mathbf{F} = \mathbf{F}\circ\Phi^{k+1} \\
    \implies& \mathbf{F}\circ\Phi^{k+1}\circ \mathbf{F}\circ\Phi^{k+1} = \mathrm{Id}.
\end{align}
This verifies that $\mathbf{F}\circ\Phi^{k+1}$ is also an involution.
\end{proof}




\end{appendices}

\section*{Data Availability Statement}

The datasets generated during and/or analysed during the current study are available in the GitHub repository, \url{https://tinyurl.com/29kz7krz}.


\bibliography{sn-bibliography}% common bib file


\end{document}